\documentclass[12pt]{article}%
\usepackage{amsfonts}
\usepackage{amssymb}
\usepackage{amsmath}
\usepackage{geometry}
\usepackage{graphicx}
\usepackage{xcolor}
\usepackage{setspace}%
\usepackage{float}
\usepackage[caption = false]{subfig}
\providecommand{\U}[1]{\protect\rule{.1in}{.1in}}
\newtheorem{theorem}{Theorem}

\newenvironment{proof}[1][Proof]{\noindent\textbf{#1.} }{\ \rule{0.5em}{0.5em}}
\usepackage{natbib}
\setcitestyle{authoryear,open={(},close={)}}

\begin{document}

\doublespacing

\begin{center} \textit{Original Article} \end{center}

{\LARGE Auto-G-Computation of Causal Effects on a Network }

{\Large Eric J. Tchetgen Tchetgen}$^1${\Large , Isabel Fulcher}$^2$
{\Large and Ilya Shpitser}$^3$ 

\begin{center}
$^1${\large Wharton Statistics Department, University of Pennsylvania} \\
$^2${\large Department of Biostatistics, Harvard University} \\
$^3${\large Department of Computer Science, Johns Hopkins University}
\end{center}

\begin{abstract}
	\noindent Methods for inferring average causal effects have traditionally relied on two key assumptions: (i) the intervention received by one unit cannot causally influence the outcome of another; and (ii) units can be organized into non-overlapping groups such that outcomes of units in separate groups are independent. In this paper, we develop new statistical methods for causal inference based on a single realization of a network of connected units for which neither assumption (i) nor (ii) holds. The proposed approach allows both for arbitrary forms of interference, whereby the outcome of a unit may depend on interventions received by other units with whom a network path through connected units exists; and long range dependence, whereby outcomes for any two units likewise connected by a path in the network may be dependent. Under network versions of consistency and no unobserved confounding, inference is made tractable by an assumption that the network’s outcome, treatment and covariate vectors are a single realization of a certain chain graph model. This assumption allows inferences about various network causal effects via the \emph{auto-g-computation algorithm}, a network generalization of Robins' well-known g-computation algorithm previously described for causal inference under assumptions (i) and (ii).
\end{abstract}


\bigskip\pagebreak

\section*{\centering 1. INTRODUCTION}

Statistical methods for inferring average causal effects in a population of
units have traditionally assumed (i) that the outcome of one unit cannot be
influenced by an intervention received by another, also known as the
no-interference assumption \citep{cox1958planning,rubin1974estimating}; and (ii) that units can
be organized into non-overlapping groups, blocks or clusters such that
outcomes of units in separate groups are independent and the number of groups
grows with sample size. Only fairly recently has causal inference literature
formally considered settings where assumption (i) does not necessarily hold 
\citep{sobel2006randomized,rosenbaum2007interference,hudgens2008toward,hong2006evaluating,graham2008identifying,manski2013identification,tchetgen2012causal}.

Early work on relaxing assumption (i) considered blocks of non-overlapping
units, where assumptions (i) and (ii) held across blocks, but not necessarily
within blocks. This setting is known as \emph{partial interference} \citep{sobel2006randomized,hong2006evaluating,hudgens2008toward,tchetgen2012causal,liu2014large,lundin2014estimation,ferracci2014evidence}.

More recent literature has sought to further relax the assumption of partial
interference by allowing the pattern of interference to be somewhat arbitrary
\citep{verbitsky2012causal,aronow2017estimating,liu2016inverse,sofrygin2017semi}, 
while still restricting a unit's set
of interfering units to be a small set defined by spacial proximity or network
ties, as well as severely limiting the degree of outcome dependence in order
to facilitate inference. A separate strand of work has primarily focused on
detection of specific forms of spillover effects in the context of an
experimental design in which the intervention assignment process is known to
the analyst \citep{aronow2012general,bowers2013reasoning,athey2018exact}. In much of
this work, outcome dependence across units can be left fairly arbitrary,
therefore relaxing (ii), without compromising validity of randomization tests
for spillover effects. Similar methods for non-experimental data, such as
observational studies, are not currently available.

Another area of research which has recently received increased interest in the
interference literature concerns the task of effect decomposition of the
spillover effect of an intervention on an outcome known to spread over a given
network into so-called contagion and infectiousness components \citep{vanderweele2012components}. 
The first quantifies the extent to which an intervention
received by one person may prevent another person's outcome from occurring
because the intervention prevents the first from experiencing the outcome and
thus somehow from transmitting it to another \citep{vanderweele2012components,ogburn2014causal,shpitser2017modeling}. 
The second quantifies the extent to which even if a person experiences the outcome, 
the intervention may impair his or her ability to transmit the outcome to another. A prominent example of
such queries corresponds to vaccine studies for an infectious disease
\citep{vanderweele2012components,ogburn2014causal,shpitser2017modeling}.
\ In this latter strand of work, it is typically assumed that interference and
outcome dependence occur only within non-overlapping groups, and that the number
of independent groups is large.

We refer the reader to \cite{tchetgen2012causal}, \cite{vanderweele2014interference}, and \cite{halloran2016dependent} for extensive overviews of the fast growing literature on interference and spillover effects.

An important gap remains in the current literature: no general approach exists
which can be used to facilitate the evaluation of spillover effects on a
single network in settings where treatment outcome relationships are
confounded, unit interference may be due not only to immediate network ties
but also from indirect connections (friend of a friend, and so on) in a
network, and non-trivial dependence between outcomes may exist for units
connected via long range indirect relationships in a network.

The current paper aims to fill this important gap in the literature.
Specifically, in this paper, the outcome experienced by a given unit could in
principle be influenced by an intervention received by a unit with whom no
direct network tie exists, provided there is a path of connected units linking
the two. Furthermore, the approach developed in this paper respects a
fundamental feature of outcomes measured on a network, by allowing for an
association of outcomes for any two units connected by a path on the network.
Although network causal effects are shown to in principle be nonparametrically
identified by a network version of the g-formula \citep{robins1986new} under standard
assumptions of consistency and no unmeasured confounding adapted to the
network setting, statistical inference is however intractable given the single
realization of data observed on the network and lack of partial interference
assumption. Nonetheless, progress is made by an assumption that network data
admit a representation as a graphical model corresponding to \emph{chain
	graphs} \citep{lauritzen2002chain}. This graphical representation of
network data generalizes that introduced in \cite{shpitser2017modeling} for the
purpose of interrogating causal effects under partial interference and it is
particularly fruitful in the setting of a single network as it implies, under
fairly mild positivity conditions, that the outcomes observed on the network
may be viewed as a single realization of a certain conditional Markov random
field (MRF); and that the set of confounders likewise constitute a single
realization of an MRF. By leveraging the local Markov property associated with
the resulting chain graph which we encode in non-lattice versions of Besag's
auto-models \citep{besag1974spatial}, we develop a certain Gibbs sampling algorithm which
we call the \emph{auto-g-computation algorithm} as a general approach to
evaluate network effects such as direct and spillover effects. \ Furthermore,
we describe corresponding statistical techniques to draw inference which
appropriately account for interference and complex outcome dependence across
the network. Auto-g-computation may be viewed as a network generalization of
Robins' well-known g-computation algorithm previously described for causal
inference under no-interference and i.i.d data \citep{robins1986new}. We also note
that while MRFs have a longstanding history as models for network data
starting with \cite{besag1974spatial} (see also \cite{kolaczyk2014statistical} for a textbook treatment
and summary of this literature), a general chain graph representation of
network data appears not to have previously been used in the context of
interference and this paper appears to be the first instance of their use in
conjunction with g-computation in a formal counterfactual framework for
inferring causal effects from observational network data. 

\cite{ogburn2017causal} have recently proposed in parallel to this work, an alternative approach for
evaluating causal effects on a single realization of a network, which is based
on traditional causal directed acyclic graphs (DAG) and their algebraic
representation as causal structural equation models. \ As discussed in
\cite{lauritzen2002chain}, such alternative representation as a DAG will
generally be incompatible with our chain graph representation and therefore
the respective contribution of these two manuscripts present little to no overlap. 
Specifically, similar to our setting,  \cite{ogburn2017causal} allow for a
single realization of the network which is fully observed; however, they
assume (i) an underlying nonparametric structural equation model with
independent error terms \citep{pearl2000causality} compatible with a certain DAG
generated the network data. This assumption implies a large number of
cross-world counterfactual independences which are largely unnecessary for
identification but inherent to their model \citep{richardson2013single}.
Furthermore, (ii) their approach precludes any dependence between outcomes
not directly connected on the network nor does it allow for interference
between units which are not network ties. Finally, (iii) inferences are
primarily based on an assumption that outcome errors for the network are
conditionally independent given baseline characteristics. Our proposed approach do not
require any of assumptions (i)-(iii).

The remainder of this paper is organized as followed. In Section 2 we present
notation used throughout. In Section 3 we review notions of direct and
spillover effects which arise in the presence of interference. In this same
section, we review sufficient conditions for identification of network causal
effects by a network version of the g-formula, assuming the knowledge of the
observed data distribution, or (alternatively) infinitely many realizations
from this distribution. We then argue that the network g-formula cannot be
empirically identified nonparametrically in more realistic settings where a
\emph{single} realization of the network is observed. To remedy this
difficulty, we leverage information encoding network ties (which we assume is
both available and accurate) to obtain a chain graph representation of
observed variables for units of the network. This chain graph is then shown to
induce conditional independences which allow versions of coding and pseudo
maximum likelihood estimators due to \cite{besag1974spatial}  to be used to make
inferences about the parameters of the joint distribution of the observed data
sample. These estimators are described in Section 4, for parametric
auto-models of \cite{besag1974spatial}. The resulting parametrization is then used to
make inferences about network causal effects via a specialized Gibbs sampling
algorithm we have called the auto-g-computation algorithm, also described in
Section 4. In Section 5, we describe results from a simulation study
evaluating the performance of the proposed approach. Finally, in Section 6, we
offer some concluding remarks and directions for future research.

\section*{\centering 2. NOTATION AND DEFINITIONS}

\subsection*{2.1 Preliminaries}

\noindent Suppose one has observed data on a population of $N$ interconnected
units. Specifically, for each $i\in\{1,\ldots N\}$ one has observed
$(A_{i},Y_{i})$, where $A_{i}$ denotes the binary treatment or intervention
received by unit $i$, and $Y_{i}$ is the corresponding outcome. Let
$\mathbf{A}\equiv(A_{1},\ldots,A_{N})$ denote the vector of
treatments all individuals received, which takes values in the set
$\{0,1\}^{N},$ and $\mathbf{A}_{-j}\equiv(A_{1},\ldots A_{N})\backslash A_{j}\equiv(A_{1},\ldots,A_{j-1},A_{j+1},\ldots A_{N})$
denote the $N-1$ subvector of $\mathbf{A}$ with the $jth$ entry deleted.
In general, for any vector $\mathbf{X=}\left(  X_{i},...,X_{N}\right)  ,$ $\mathbf{X}_{-j}=(X_{1},...,X_{N})\backslash
X_{j}=(X_{1},...,X_{j-1},X_{j+1},...,X_{N}).$ Likewise if
$X_{i}=(X_{1,i},...,X_{p,i})$ is a vector with $p$ components, $X_{\backslash
	s,i}=(X_{1,i},...,X_{s-1,i},X_{s+1,i},...,X_{p,i}).$ Following \cite{sobel2006randomized}
and \cite{hudgens2008toward}, we refer to $\mathbf{A}$ as an intervention,
treatment or allocation program, to distinguish it from the individual
treatment $A_{i}.\ $Furthermore, for $n=1,2,\ldots,$ we define $\mathcal{A(}%
n)$ as the set of vectors of possible treatment allocations of length $n$; for
instance $\mathcal{A(}2)\equiv\left\{  (0,0),(0,1),(1,0),(1,1)\right\}  .$
Therefore, $\mathbf{A}$ takes one of $2^{N}$ possible values in $\mathcal{A(}%
N)$, while $\mathbf{A}_{-j}$ takes values in $\mathcal{A(}N-1)$ for all $j$. \ 

As standard in causal inference, we assume the existence of counterfactual
(potential outcome) data $\mathbf{Y}(\mathbf{\cdot})=\{Y_{i}(\mathbf{a}%
):\mathbf{a}\in\mathcal{A(}N)\mathcal{\}}$ where $\mathbf{Y}(\mathbf{a}%
)=\{Y_{1}\left(  \mathbf{a}\right)  ,\ldots,Y_{N}(\mathbf{a})\}$,
$Y_{i}\left(  \mathbf{a}\right)  $ is unit $i^{\prime}s$ response under
treatment allocation $\mathbf{a}$; and that the observed outcome $Y_{i}$ for
unit $i$ is equal to his counterfactual outcome $Y_{i}\left(  \mathbf{A}%
\right)  $ under the realized treatment allocation $\mathbf{A;}$ more
formally, we assume the network version of the consistency assumption in
causal inference:
\begin{equation}
\mathbf{Y}\left(  \mathbf{A}\right)  =\mathbf{Y}\text{ a.e.} \label{NetCons}%
\end{equation}
Notation for the random variable $Y_{i}(\mathbf{a})$ makes explicit the
possibility of the potential outcome for unit $i$ depending on treatment
values of other units, that is the possibility of interference. The standard
no-interference assumption \citep{cox1958planning,rubin1974estimating} made in the causal
inference literature, namely that for all $j$ if $\mathbf{a}$ and
$\mathbf{a}^{\prime}$ are such that $a_{j}=a_{j}^{\prime}$ then $Y_{j}\left(
\mathbf{a}\right)  =Y_{j}\left(  \mathbf{a}^{\prime}\right)  $ a.e., implies
that the counterfactual outcomes for individual $j$ can be written in a
simplified form as $\left\{  Y_{j}\left(  a\right)  :a\in\{0,1\}\right\}  $.
The partial interference assumption \citep{sobel2006randomized,hudgens2008toward,tchetgen2012causal}, which weakens the no-interference
assumption, assumes that the $N$ units can be partitioned into $K$ blocks of
units, such that interference may occur within a block but not between blocks.
Under partial interference, $Y_{i}\left(  \mathbf{a}\right)  =Y_{i}\left(
\mathbf{a}^{\prime}\right)  $ a.s. only if $a_{j}=a_{j}^{\prime}$ for all $j$
in the same block as unit $i.$ The assumption of partial interference is
particularly appropriate when the observed blocks are well separated by space
or time such as in certain group randomized studies in the social sciences, or
community-randomized vaccine trials. \cite{aronow2017estimating} relaxed the
requirement of non-overlapping blocks, and allowed for more complex patterns
of interference across the network. Obtaining identification required a priori
knowledge of the \textquotedblleft interference set,\textquotedblright\ that
is for each unit $i$, the knowledge of the set of units $\left\{
j:Y_{i}\left(  \mathbf{a}\right) \neq Y_{i}\left(  \mathbf{a}^{\prime}\right)
\text{ a.s. if }a_{k}=a_{k}^{\prime}\text{ and }a_{j}\not =a_{j}^{\prime
}\text{ for all }k\neq j\right\}  $. In addition, the number of units
interfering with any given unit had to be negligible relative to the size of
the network. See \cite{liu2016inverse} for closely related assumptions.

In contrast to existing approaches, our approach allows \emph{full} rather
than partial interference in settings where treatments are also not
necessarily randomly assigned. The assumptions that we make can be separated
into two parts: network versions of standard causal inference assumptions,
given below, and independence restrictions placed on the observed data
distribution which can be described by a graphical model, described in more
detail later.

We assume that for each $\mathbf{a}\in\mathcal{A(}N)$ the vector of potential
outcomes $\mathbf{Y}(\mathbf{a})$ is a single realization of a random field.
In addition to treatment and outcome data, we suppose that one has also
observed a realization of a (multivariate) random field $\mathbf{L}=\left(
L_{1},\ldots,L_{N}\right)  ,$ where $L_{i}$ denotes pre-treatment
covariates for unit $i$. For identification purposes, we take advantage of a
network version of the conditional ignorability assumption about treatment
allocation which is analogous to the standard assumption often made in causal
inference settings; specifically, we assume that:
\begin{equation}
\mathbf{A}\perp\!\!\!\perp\mathbf{Y}(\mathbf{a})|\mathbf{L}\text{ for all
}\mathbf{a}\in\mathcal{A(}N),\text{ } \label{NetIgn}%
\end{equation}
This assumption basically states that all relevant information used in
generating the treatment allocation whether by a researcher in an experiment
or by "nature" in an observational setting, is contained in $\mathbf{L.}$
Network ignorability can be enforced in an experimental design where treatment
allocation is under the researcher's control. On the other hand, the
assumption cannot be ensured to hold in an observational study since treatment
allocation is no longer under experimental control, in which case credibility
of the assumption depends crucially on subject matter grounds. Equation
$\left(  \ref{NetIgn}\right)  $ simplifies to the standard assumption of no
unmeasured confounding in the case of no interference and i.i.d. unit data, in
which case $A_{i}\perp\!\!\!\perp Y_{i}\left(  a\right)  |L_{i}$ for all
$a\in\left\{  0,1\right\}  $. $\ $Finally, we make the following positivity
assumption at the network treatment allocation level:%
\begin{equation}
f\left(  \mathbf{a}|\mathbf{L}\right)  >\sigma>0\text{ a.e. for all
}\mathbf{a}\in\mathcal{A(}N). \label{positivity}%
\end{equation}
\bigskip

\subsection*{2.2 Network causal effects}

We will consider a variety of network causal effects that are expressed in
terms of unit potential outcome expectations $\psi_{i}\left(  \mathbf{a}%
\right)  =E\left(  Y_{i}\left(  \mathbf{a}\right)  \right)  ,$ $i=1,...,N.$
Let $\psi_{i}\left(  \mathbf{a}_{-i},a_{i}\right)  =E\left(  Y_{i}\left(
\mathbf{a}_{-i},a_{i}\right)  \right)  $ The following definitions are
motivated by analogous definitions for fixed counterfactuals given in \cite{hudgens2008toward}. 
The first definition gives the average direct causal
effect for unit $i$ upon changing the unit's treatment status from inactive
($a=0)$ to active ($a=1)$ while setting the treatment received by other units
to $\mathbf{a}_{-i}:$
\[
DE_{i}\left(  \mathbf{a}_{-i}\right)  \equiv\psi_{i}\left(  \mathbf{a}%
_{-i},a_{i}=1\right)  -\psi_{i}\left(  \mathbf{a}_{-i},a_{i}=0\right)  ;
\]
The second definition gives the average spillover (or ``indirect") causal effect experienced by
unit $i$ upon setting the unit's treatment inactive, while changing the
treatment of other units from inactive to $\mathbf{a}_{-i}:$%
\[
IE_{i}\left(  \mathbf{a}_{-i}\right)  \equiv\psi_{i}\left(  \mathbf{a}%
_{-i},a_{i}=0\right)  -\psi_{i}\left(  \mathbf{a}_{-i}=\mathbf{0},a_{i}=0\right)  ;
\]
Similar to \cite{hudgens2008toward} these effects can be averaged over a
hypothetical allocation regime $\pi_{i}\left(  \mathbf{a}_{-i};\alpha\right)
$ indexed by $\alpha$ to obtain allocatio\text{n-specific unit average direct
	and spillover effects} $DE_{i}\left(  \alpha\right)  =\sum_{\mathbf{a}_{-i}%
	\in\mathcal{A(}N)}\pi_{i}\left(  \mathbf{a}_{-i};\alpha\right)  DE_{i}\left(
\mathbf{a}_{-i}\right)  $ \text{ and }$IE_{i}\left(  \alpha\right)
=\sum_{\mathbf{a}_{-i}\in\mathcal{A(}N)}\pi_{i}\left(  \mathbf{a}_{-i}%
;\alpha\right)  IE_{i}\left(  \mathbf{a}_{-i}\right)  ,$ respectively. One may
further average over the units in the network to obtain allocation-specific
network average direct and spillover effects $DE\left(  \alpha\right)  =$
$N^{-1}\sum_{i}DE_{i}\left(  \alpha\right)  $ and $IE\left(  \alpha\right)  =$
$N^{-1}\sum_{i}IE_{i}\left(  \alpha\right)  $, respectively. These quantities
can further be used to obtain other related network effects such as average
total and overall effects at the unit or network level analogous to \cite{hudgens2008toward}
and \cite{tchetgen2012causal}.

Identification of these effects follow from identification of $\psi_{i}\left(
\mathbf{a}\right)  $ for each $i=1,...,N.$ In fact, under assumptions $\left(
\ref{NetCons}\right)  $-$\left(  \ref{positivity}\right)  ,$ it is
straightforward to show that $\psi_{i}\left(  \mathbf{a}\right)  $ is given by
a network version of Robins' g-formula: $\psi_{i}\left(  \mathbf{a}\right)
=\beta_{i}\left(  \mathbf{a}\right)  $ where $\beta_{i}\left(  \mathbf{a}%
\right)  \equiv\sum_{\mathbf{l}}E\left(  Y_{i}|\mathbf{A=a,L=l}\right)
f\left(  \mathbf{l}\right)  \mathbf{,}$ $f\left(  \mathbf{l}\right)  $ is the
density of $\mathbf{l,}$ and $\sum$ may be interpreted as integral when appropriate.

Although $\psi_{i}\left(  \mathbf{a}\right)  $ can be expressed as the
functional $\beta_{i}\left(  \mathbf{a}\right)  $ of the observed data law,
$\beta_{i}\left(  \mathbf{a}\right)  $ cannot be identified nonparametrically
from \emph{a single realization} $(\mathbf{Y,A,L)}$ drawn from this law
without imposing additional assumptions. In the absence of interference, it is
standard to rely on the additional assumption that $(Y_{i},A_{i}%
,L_{i}\mathbf{)}$, $i=1,...N~$are i.i.d., in which case the above g-formula
reduces to the standard g-formula $\beta_{i}\left(  \mathbf{a}\right)
=\beta\left(  a_{i}\right)  =\sum_{l}E\left(  Y_{i}|A_{i}=a_{i},L_{i}%
=l\right)  f(l\mathbf{)\,\ }$ which is nonparametrically identified \citep{robins1986new}. 
Since we consider a sample of interconnected units in a network, the
i.i.d. assumption is unrealistic. Below, we consider assumptions on the
observed data law that are much weaker, but still allow inferences about
network effects to be made.

We first introduce a convenient representation of $E\left(  Y_{i}%
|\mathbf{A=a,L=l}\right)  $, and describe a corresponding Gibbs sampling
algorithm which could in principle be used to compute the network g-formula
under the unrealistic assumption that the observed data law is known. First,
note that $\beta_{i}\left(  \mathbf{a}\right)  =\sum_{\mathbf{y,l}}%
y_{i}f\left(  \mathbf{y|A=a,L=l}\right)  f\left(  \mathbf{l}\right)  .$

Suppose that one has available the conditional densities (also referred to as
Gibbs factors) $f\left(  Y_{i}\mathbf{|Y}_{-i}=\mathbf{y}_{-i},\mathbf{a,l}%
\right)  $ and $f\left(  L_{i}\mathbf{|L}_{-i}\mathbf{=l}_{-i}\right)  $,
$i=1,...,N,\,$\ and that it is straightforward to sample from these densities.
Then, evaluation of the above formula for $\beta_{i}\left(  \mathbf{a}\right)
$ can be achieved with the following Gibbs sampling algorithm.

\underline{Gibbs Sampler I}:%

\begin{align*}
\text{for }m  &  =0,\text{let }\left(  \mathbf{L}^{(0)},\mathbf{Y}%
^{(0)}\right)  \text{ denote initial values ;}\\
\text{for }m  &  =0,...,M\\
&  \text{let }i=(m\mod N)+1;\\
&  \text{draw }L_{i}^{(m+1)}\text{ from }f\left(  L_{i}\mathbf{|L}_{-i}%
^{(m)}\right)  \text{ and }Y_{i}^{(m+1)}\text{ from }f\left(  Y_{i}%
\mathbf{|Y}_{-i}^{(m)},\mathbf{a,L}^{(m)}\right)  ;\\
&  \text{let }\mathbf{L}_{-i}^{(m+1)}\left.  =\right.  \mathbf{L}_{-i}%
^{(m)}\text{ and }\mathbf{Y}_{-i}^{(m+1)}\left.  =\right.  \mathbf{Y}%
_{-i}^{(m)}.\\
&
\end{align*}
The sequence $\left(  \mathbf{L}^{(0)},\mathbf{Y}^{(0)}\right)  ,\left(
\mathbf{L}^{(1)},\mathbf{Y}^{(1)}\right)  ,\ldots,\left(  \mathbf{L}%
^{(m)},\mathbf{Y}^{(m)}\right)  $ forms a Markov chain, which under
appropriate regularity conditions converges to the stationary distribution
$f\left(  \mathbf{Y|a,L}\right)  \times f\left(  \mathbf{L}\right)  $ \citep{liu2008monte}. 
Specifically, we assume $M$ is an integer larger than the number of
transitions necessary for the appropriate Markov chain to reach equilibrium
from the starting state. \ Thus, for sufficiently large $m^{\ast}$ and $K,$%
\[
\beta_{i}\left(  \mathbf{a}\right)  \approx K^{-1}\sum_{k=0}^{K}%
Y_{i}^{(m^{\ast}+k)}.
\]
Thus, if Gibbs factors $f\left(  Y_{i}\mathbf{|Y}_{-i}=\mathbf{y}%
_{-i},\mathbf{a,l}\right)  $ and $f\left(  L_{i}\mathbf{|L}_{-i}%
\mathbf{=l}_{-i}\right)  $ are available for every $i$, all networks causal
effects can be computed. This approach to evaluating the g-formula is the
network analogue of Monte Carlo sampling approaches to evaluating functionals
arising from the g-computation algorithm in the sequentially ignorable model,
see for instance \cite{westreich2012parametric}.
%{\color{red}[added.-i]}
Unfortunately these factors are not identified from a single realization of
the observed data law, without additional assumptions. In the following
section we describe additional assumptions which will imply identification.

\section*{\centering 3. A GRAPHICAL STATISTICAL MODEL FOR NETWORK DATA}

To motivate our approach, we introduce a representation for network data
proposed by \cite{shpitser2017modeling} and based on chain graphs. A
chain graph (CG) \citep{lauritzen1996graphical} is a mixed graph containing undirected
($-$) and directed ($\to$) edges with the property that it is impossible to
add orientations to undirected edges in such a way as to create a directed
cycle. A chain graph without undirected edges is called a directed acyclic
graph (DAG).

A statistical model associated with a CG $\mathcal{G}$ with a vertex set
$\mathbf{O}$ is a set of densities that obey the following two level
factorization:
\begin{align}
p(\mathbf{O}) = \prod_{\mathbf{B} \in\mathcal{B}(\mathcal{G})} p(\mathbf{B}
\mid\text{pa}_{\mathcal{G}}(\mathbf{B})),\label{eqn:fact1}%
\end{align}
where $\mathcal{B}(\mathcal{G})$ is the partition of vertices in $\mathcal{G}$
into \emph{blocks}, or sets of connected components via undirected edges, and
$\text{pa}_{\mathcal{G}}(\mathbf{B})$ is the set $\{ W : W \to B
\in\mathbf{B} \text{ exists in }\mathcal{G} \}$. This outer factorization
resembles the Markov factorization of DAG models. Furthermore, each factor
$p(\mathbf{B} \mid\text{pa}_{\mathcal{G}}(\mathbf{B}))$ obeys the following
inner factorization, which is a clique factorization for a conditional Markov
random field:
\begin{align}
p(\mathbf{B} \mid\text{pa}_{\mathcal{G}}(\mathbf{B})) = \frac{1}%
{Z(\text{pa}_{\mathcal{G}}(\mathbf{B}))} \prod_{\mathbf{C} \in\mathcal{C}%
	(\mathcal{G}^{a}_{\mathbf{B} \cup\text{pa}_{\mathcal{G}}(\mathbf{B}))});
	\mathbf{C} \not \subseteq \text{pa}_{\mathcal{G}}(\mathbf{B})} \phi
_{\mathbf{C}}(\mathbf{C}),\label{eqn:fact2}%
\end{align}
where $Z(\text{pa}_{\mathcal{G}}(\mathbf{B}))$ is a normalizing function which
ensures a valid conditional density, $\mathcal{C}(\mathcal{G})$ is a set of
maximal pairwise connected components (cliques) in an undirected graph
$\mathcal{G}$, $\phi_{\mathcal{C}}(\mathbf{C})$ is a mapping from values of
$\mathbf{C}$ to real numbers, and $\mathcal{G}^{a}_{\mathbf{B} \cup
	\text{pa}_{\mathcal{G}}(\mathbf{B}))}$ is an undirected graph with vertices
$\mathbf{B} \cup\text{pa}_{\mathcal{G}}(\mathbf{B})$ and an edge between any
pair in $\text{pa}_{\mathcal{G}}(\mathbf{B})$ and any pair in $\mathbf{B}
\cup\text{pa}_{\mathcal{G}}(\mathbf{B})$ adjacent in $\mathcal{G}$.

A density $p(\mathbf{O})$ that obeys the two level factorization given by
(\ref{eqn:fact1}) and (\ref{eqn:fact2}) with respect to a CG $\mathcal{G}$ is
said to be Markov relative to $\mathcal{G}$. This factorization implies a
number of Markov properties relating conditional independences in
$p(\mathbf{O})$ and missing edges in $\mathcal{G}$. Conversely, these Markov
properties imply the factorization under an appropriate version of the
Hammersley-Clifford theorem, which does not hold for all densities, but does
hold for wide classes of densities, which includes positive densities \citep{hammersley1971markov}. Special
cases of these Markov properties are described further below. Details can be
found in \cite{lauritzen1996graphical}.

\subsection*{3.1 A chain graph representation of network data}

Observed data distributions entailed by causal models of a DAG do not
necessarily yield a good representation of network data. This is because DAGs
impose an ordering on variables that is natural in temporally ordered
longitudinal studies but not necessarily in network settings. As we now show
the Markov property associated with CGs accommodates both dependences
associated with causal or temporal orderings of variables, but also symmetric
dependences induced by the network.

Let $\mathcal{E}$ denote the set of neighboring pairs of units in the network;
that is $(i,j)\in\mathcal{E}$ only if units $i$ and $j$ are directly connected
on the network. We represent data $\mathbf{O}$ drawn from a joint distribution
associated with a network with neighboring pairs $\mathcal{E}$ as a CG
$\mathcal{G}_{\mathcal{E}}$ in which each variable corresponds to a vertex,
and directed and undirected edges of $\mathcal{G}_{\mathcal{E}}$ are defined
as follows. For each pair of units $(i,j)\in\mathcal{E}$, variables $L_{i}$
and $L_{j}$ are connected by an undirected edge in $\mathcal{G}_{\mathcal{E}}%
$. We use an undirected edge to represent the fact that $L_{i}$ and $L_{j}$
are associated, but this association is not in general due to unobserved
common causes, nor as the variables are contemporaneous can they be ordered
temporally or causally \citep{shpitser2017modeling}. Vertices for
$A_{i}$ and $A_{j},$ and $Y_{i}$ and $Y_{j}\,$\ are likewise connected by an
undirected edge in $\mathcal{G}_{\mathcal{E}}$ if and only if $(i,j)\in$
$\mathcal{E}$. Furthermore, for each $(i,j)\in$ $\mathcal{E}$, a directed edge
connects $L_{i}$ to both $A_{i}$ and $A_{j}$ encoding the fact that covariates
of a given unit may be direct causes of the unit's treatment but also of the
neighbor treatments, i.e. $L_{i}\rightarrow$ $\left\{  A_{i},A_{j}\right\}  ;$
edges $L_{i}\rightarrow$ $\left\{  Y_{i},Y_{j}\right\}  $ and $A_{i}%
\rightarrow$ $\left\{  Y_{i},Y_{j}\right\}  $ should be added to the chain
graph for a \ similar reason. As an illustration, the CG in Figure 1 corresponds 
to a three-unit network where $\mathcal{E=}\left\{  \left(  1,2\right)  ,\left(  2,3\right)  \right\}  $.

\begin{figure}[ptbh] %TO DO: Update chain graph
	\caption{Chain graph representation of data from a network of three units}%
	\centering
	\includegraphics[scale=.3]{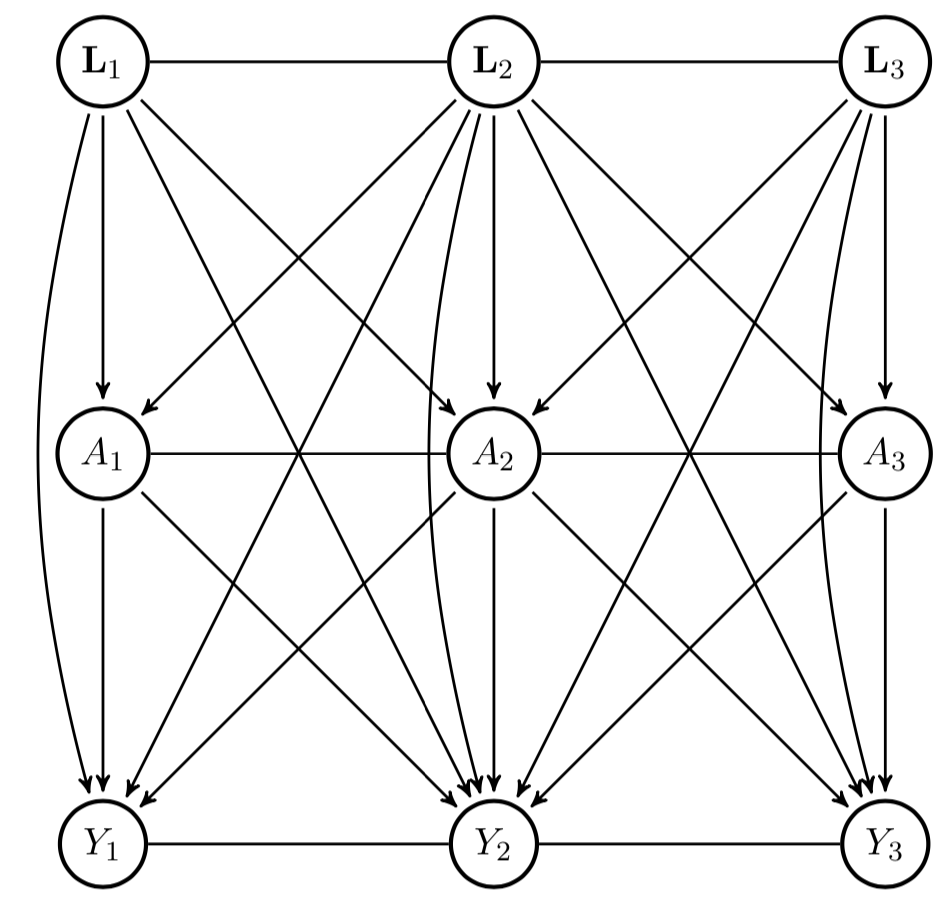}\end{figure}

We will assume the observed data distribution on $\mathbf{O}$ associated with
our network causal model is Markov relative to the CG constructed from unit
connections in a network via the above two level factorization \citep{lauritzen1996graphical}. 
This implies the observed data distribution obeys certain conditional
independence restrictions that one might intuitively expect to hold in a
network, and which serve as the basis of the proposed approach. Let
$\mathcal{N}_{i}$ denote the set of neighbors of unit $i,$ i.e. $\mathcal{N}%
_{i}=\left\{  j:\left(  i,j\right)  \in\mathcal{E}\right\}  $, and let
$\mathcal{O}_{i}=\left\{  \mathbf{O}_{j},j\in\mathcal{N}_{i}\right\}  $ denote
data observed on all neighbors of unit $i.$ Given a CG $\mathcal{G}%
_{\mathcal{E}}$ with associated neighboring pairs $\mathcal{E}$, the
following conditional independences follow by the global Markov property
associated with CGs \citep{lauritzen1996graphical}:%
\begin{align}
Y_{i}  &  \perp\!\!\!\perp\{Y_{k},A_{k},L_{k}\}|(A_{i},L_{i},\mathcal{O}%
_{i})\ \text{for all }i\text{ and }k, \text{ }k\neq
i;\text{ }\label{Markov}\\
\text{ }L_{i}  &  \perp\!\!\!\perp\mathbf{L}_{-i}\setminus\mathcal{O}%
_{i}|\mathbf{L}_{-i}\cap\mathcal{O}_{i}\text{ for all }i\text{ and }%
k\notin\mathcal{N}_{i},\text{ }k\neq i. \label{Markovii}%
\end{align}
In words, equation (\ref{Markov}) states that the outcome of a given unit can
be screened-off (i.e. made independent) from the variables of all
non-neighboring units by conditioning on the unit's treatment and covariates
as well as on all data observed on its neighboring units, where the
neighborhood structure is determined by $\mathcal{G}_{\mathcal{E}}$. That is
$(A_{i},L_{i},\mathcal{O}_{i})$ is the \emph{Markov blanket }of $Y_{i}$ in CG
$\mathcal{G}_{\mathcal{E}}$. This assumption, coupled with a sparse network
structure leads to extensive dimension reduction of the model specification
for $\mathbf{Y|A,L}$. In particular, the conditional density of $Y_{i}%
|\left\{  \mathbf{O\backslash}Y_{i}\right\}  $ only depends on $\left(
A_{i},L_{i}\right)  $ and on neighbors' data $\mathcal{O}_{i}.$ Similarly,
$\mathbf{L}_{-i}\cap\mathcal{O}_{i}$ is the Markov blanket of $L_{i}$ in CG
$\mathcal{G}_{\mathcal{E}}$.

\subsection*{3.2 Conditional auto-models}

Suppose that instead of $\left(  \ref{positivity}\right)  $, the following
stronger positivity condition holds:
\begin{equation}
\mathbb{P}\left(  \mathbf{O=o}\right)  >0,\text{ for all possible values
}\mathbf{o.} \label{positivie}%
\end{equation}

Since $\left(  \ref{Markov}\right)  $ holds for the conditional law of
$\mathbf{Y}$ given $\mathbf{A},\mathbf{L}$, it lies in the conditional MRF
(CMRF) model associated with the induced undirected graph $\mathcal{G}%
_{\mathcal{E}}^{a}$.
%(\mathbf{Y},\mathbf{A}\cup\mathbf{L})$.
In addition, since $\left(  \ref{positivie}\right)  $ holds, the conditional
MRF version of the Hammersley-Clifford (H-C) theorem and
$\left(  \ref{Markov}\right)  $ imply the following version of the clique factorization in (\ref{eqn:fact2}),
\[
f\left(  \mathbf{y|a,l}\right)  =\left(  \frac{1}{\kappa\left(  \mathbf{a,l}%
	\right)  }\right)  \exp\left\{  U\left(  \mathbf{y;a,l}\right)  \right\}  ,
\]
where $\kappa\left(  \mathbf{a,l}\right)  =\sum_{\mathbf{y}}\exp\left\{
U\left(  \mathbf{y;a,l}\right)  \right\}  ,$ and $U\left(  \mathbf{y;a,l}%
\right)  $ is a \emph{conditional energy function} which can be decomposed
into a sum of terms called conditional clique potentials, with a term for
every maximal clique in the graph $\mathcal{G}_{\mathcal{E}}^{a}$ \citep{besag1974spatial}.Conditional
clique potentials offer a natural way to specify a CMRF using only terms that
depend on a small set of variables. Specifically,
\begin{align}
&  f\left(  Y_{i}=y_{i}|\mathbf{Y}_{-i}=\mathbf{y}_{-i},\mathbf{a,l}\right)
\nonumber\\
&  =f\left(  Y_{i}=y_{i}|\mathbf{Y}_{-i}=\mathbf{y}_{-i},\left\{
a_{j}\mathbf{,}l_{j}:j\in\mathcal{N}_{i}\right\}  \text{ }\right) \nonumber\\
&  =\frac{\exp\left\{  \sum_{c\in\mathcal{C}_{i}}U_{c}\left(  \mathbf{y;a,l}%
	\right)  \right\}  }{\sum_{\mathbf{y}^{\prime}\mathbf{:y}_{-i}^{\prime
		}=\mathbf{y}_{-i}}\exp\left\{  \sum_{c\in\mathcal{C}_{i}}U_{c}\left(
	\mathbf{y}^{\prime}\mathbf{;a,l}\right)  \right\}  }, \label{gibbs factor}%
\end{align}
where $\mathcal{C}_{i}$ are all maximal cliques of $\mathcal{G}_{\mathcal{E}%
}^{a}$ that involve $Y_{i}$.

Gibbs densities specified as in $\left(  \ref{gibbs factor}\right)  $ is a
rich class of densities, and are often regularized in practice by setting to
zero conditional clique potentials for cliques of size greater than a
pre-specified cut-off. This type of regularization corresponds to setting
higher order interactions terms to zero in log-linear models. For instance,
closely following \cite{besag1974spatial}, one may introduce conditions (a) only cliques
$c\in\mathcal{C}$ of size one or two have non-zero potential functions
$U_{c},$ and (b) the conditional probabilities in $\left(  \ref{gibbs factor}%
\right)  $ have an exponential family form. Under these additional conditions,
given $\mathbf{a,l,}$ the energy function takes the form
\[
U\left(  \mathbf{y;a,l}\right)  =\sum_{i\in\mathcal{G}_{\mathcal{E}}}%
y_{i}G_{i}\left(  y_{i}\mathbf{;a,l}\right)  +\sum_{\left\{  i,j\right\}
	\in\mathcal{E}}y_{i}y_{j}\theta_{ij}\left(  \mathbf{a,l}\right)  ,
\]
for some functions $G_{i}\left(  \cdot\mathbf{;a,l}\right)  $ and coefficients
$\theta_{ij}\left(  \mathbf{a,l}\right)  .$ Note that in order to be
consistent with local Markov conditions $\left(  \ref{Markov})\text{ and
	(}\ref{Markovii}\right)  ,G_{i}\left(  \cdot\mathbf{;a,l}\right)  $ can only
depend on $\left\{  \left(  a_{s},l_{s}\right)  :s\in\mathcal{N}_{j}\right\}
,$ while because of symmetry $\theta_{ij}\left(  \mathbf{a,l}\right)
$\textbf{ }can depend at most on $\left\{  \left(  a_{s},l_{s}\right)
:s\in\mathcal{N}_{j}\cap\mathcal{N}_{i}\right\}  $. Following \cite{besag1974spatial}, we
call the resulting class of models \emph{conditional auto-models}.

Conditions $\left(  \ref{Markovii}\right)  $ and $\left(  \ref{positivie}%
\right)  $ imply that $\mathbf{L}$ is an MRF; standard Hammersley-Clifford
theorem further implies that the joint density of $\mathbf{L}$ can be written
as
\[
f\left(  \mathbf{l}\right)  =\left(  \frac{1}{\nu}\right)  \exp\left\{
W\left(  \mathbf{l}\right)  \right\}
\]
where $\nu=\sum_{\mathbf{l}^{\prime}}\exp\left\{  W\left(  \mathbf{l}^{\prime
}\right)  \right\}  $, and $W\left(  \mathbf{l}\right)  $ is an energy
function which can be decomposed as a sum over cliques in the induced
undirected graph $(\mathcal{G}_{\mathcal{E}})_{\mathbf{L}}$. Analogous to the
conditional auto-model described above, we restrict attention to densities of
$\mathbf{L}$ of the form:
\begin{align}
W\left(  \mathbf{L}\right)  =\sum_{i\in\mathcal{G}_{\mathcal{E}}}\left\{
\sum_{k=1}^{p}L_{k,i}H_{k,i}\left(  L_{k,i}\right)  +\sum_{k\neq s}%
\rho_{k,s,i}L_{k,i}L_{s,i}\right\}  +\sum_{\left\{  i,j\right\}
	\in\mathcal{E}}\sum_{k=1}^{p}\sum_{s=1}^{p}\omega_{k,s,i,j}L_{k,i}L_{s,j},
\label{eqn:l-model}%
\end{align}
for some functions $H_{k,i}\left(  L_{k,i}\right)  $ and coefficients
$\rho_{k,s,i},\omega_{k,s,i,j}.$ Note that $\rho_{k,s,i}$ encodes the
association between covariate $L_{k,i}$ and covariate $L_{s,i}$ observed on
unit $i,$ while $\omega_{k,s,i,j}$ captures the association between $L_{k,i}$
observed on unit $i$ and $L_{s,j}$ observed on unit $j.$

\subsection*{3.3 Parametric specifications of auto-models}

A prominent auto-regression model for binary outcomes is the so-called
auto-logistic regression first proposed by \cite{besag1974spatial}. Note that as $\left(
\mathbf{a,l}\right)  $ is likely to be high dimensional, identification and
inference about $G_{i}$ and $\theta_{ij}$ requires one to further restrict
heterogeneity by specifying simple low dimensional parametric models for these
functions of the form$:$
\begin{align*}
G_{i}\left(  y_{i}\mathbf{;a,l}\right)   &  =\widetilde{G}_{i}\left(
\mathbf{a,l}\right)  =\mathrm{log}\frac{\Pr\left(  Y_{i}=1|\mathbf{a,l,Y}%
	_{-i}=0\right)  }{\Pr\left(  Y_{i}=0|\mathbf{a,l,Y}_{-i}=0\right)  }\\
&  =\beta_{0}+\beta_{1}a_{i}+\beta_{2}^{\prime}l_{i}+\beta_{3}\sum
_{j\in\mathcal{N}_{i}}w_{ij}^{a}a_{j}+\beta_{4}^{\prime}\sum_{j\in
	\mathcal{N}_{i}}w_{ij}^{l}l_{j};\\
\theta_{ij}  &  =w_{ij}^{y}\theta,
\end{align*}
where $w_{ij}^{a}$, $w_{ij}^{l}$, $w_{ij}^{y}$ are user specified weights
which may depend on network features associated with units $i$ and $j$, with
$\sum_{j}w_{ij}^{a}=\sum_{j}w_{ij}^{l}=\sum_{j}w_{ij}^{y}=1;$ e.g. $w_{ij}%
^{a}=1/\mathrm{card}\left(  \mathcal{N}_{i}\right)  $ standardizes the
regression coefficient by the size of a unit's neighborhood. We assume model
parameters $\tau=\left(  \beta_{0},\beta_{1},\beta_{2}^{\prime},\beta
_{3},\beta_{4}^{\prime},\theta\right)  $ are shared across units in a network. 
In addition, network features can be incorporated into the auto-models as model parameters, which may be desirable in settings where network features are confounders for the relationship between exposure and outcome. For example, one could further adjust for a unit's degree (i.e. number of ties). 

For a continuous outcome, an auto-Gaussian model may be specified as
followed:
\begin{align*}
G_{i}\left(  y_{i}\mathbf{;a,l}\right)   &  =-\left(  \frac{1}{2\sigma_{y}%
	^{2}}\right)  (y_{i}-2\mu_{y,i}\left(  \mathbf{a,l}\right)  );\\
\mu_{y,i}\left(  \mathbf{a,l}\right)   &  =\beta_{0}+\beta_{1}a_{i}+\beta
_{2}^{\prime}l_{i}+\beta_{3}\sum_{j\in\mathcal{N}_{i}}w_{ij}^{a}a_{j}%
+\beta_{4}^{\prime}\sum_{j\in\mathcal{N}_{i}}w_{ij}^{l}l_{j};\\
\theta_{ij}  &  =w_{ij}^{y}\theta,
\end{align*}
where $\mu_{y,i}\left(  \mathbf{a,l}\right)  =E\left(  Y_{i}|\mathbf{a,l,Y}%
_{-i}=0\right)  $, and $\sigma_{y}^{2}=\mathrm{var}\left(  Y_{i}%
|\mathbf{a,l,Y}_{-i}=0\right)  $. Similarly, model parameters $\tau
_{Y}=\left(  \beta_{0},\beta_{1},\beta_{2}^{\prime},\beta_{3},\beta
_{4}^{\prime},\sigma_{y}^{2},\theta\right)  $ are shared across units in the
network. Other auto-models within the exponential family can likewise be
conditionally specified, e.g. the auto-Poisson model.

Auto-model density of $\mathbf{L}$ is specified similarly. For example, fix
parameters in (\ref{eqn:l-model})
\begin{align*}
\rho_{k,s,i}  &  =\rho_{k,s},\\
\omega_{k,s,i,j}  &  =\widetilde{\omega}_{k,s}v_{i,j},
\end{align*}
where $v_{i,j}$ is a user-specified weight which satisfies $\sum_{j}v_{i,j}%
=1$. For $L_{k}$ binary, one might take%
\[
H_{k,i}\left(  L_{k,i};\tau_{k}\right)  =\tau_{k}=\mathrm{log}\frac
{\Pr\left(  L_{k,i}=1|L_{\backslash k,i}=0\mathbf{,L}_{-i}=0\right)  }%
{\Pr\left(  L_{k,i}=0|L_{\backslash k,i}=0\mathbf{,L}_{-i}=0\right)  },
\]
corresponding to a logistic auto-model for $L_{k,i}|L_{\backslash
	k,i}=0\mathbf{,L}_{-i}=0,$ while for continuous $L_{k}$%
\[
H_{k,i}\left(  L_{k,i};\tau_{k}=\left(  \sigma_{k}^{2},\mu_{k}\right)
\right)  =-\left(  \frac{1}{2\sigma_{k}^{2}}\right)  (L_{k,i}-2\mu_{k}),
\]
corresponding to a Gaussian auto-model for $L_{k,i}|L_{\backslash
	k,i}=0\mathbf{,L}_{-i}=0.$ As before, model parameters $\tau_{L}=(\tau
_{1}^{\prime},...,\tau_{p}^{\prime})$ are shared across units in the network.

\subsection*{3.4 Coding estimators of auto-models}

Suppose that one has specified auto-models for $\mathbf{Y}$ and $\mathbf{L}$
as in the previous section with unknown parameters $\tau_{Y}$ and $\tau_{L}$
respectively. To estimate these parameters, one could in principle attempt to
maximize the corresponding joint likelihood function. However, such task is
well-known to be computationally daunting as it requires a normalization step
which involves evaluating a high dimensional sum or integral which, outside
relatively simple auto-Gaussian models is generally not available in closed
form. For example, to evaluate the conditional likelihood of $\mathbf{Y|A,L}$
for binary $Y$ requires evaluating a sum of $2^{N}$ terms in order to compute
$\kappa\left(  \mathbf{A,L}\right)  \mathbf{.}$ Fortunately, less
computationally intensive strategies for estimating auto-models exist
including pseudo-likelihood estimation and so called-coding estimators \citep{besag1974spatial},
which may be adopted here. We first consider \emph{ coding-type
	estimators}, mainly because unlike pseudo-likelihood estimation, standard
asymptotic theory applies. To describe these estimators in more detail
requires additional definitions.

We define a \emph{stable set} or \emph{independent set} on $\mathcal{G}_{\mathcal{E}}$ as the
set of nodes, $\mathcal{S}\left(\mathcal{G}_{\mathcal{E}}\right)$, such that 
$$ (i,j) \notin \mathcal{E} \ \forall (i,j) \in \mathcal{S}\left(\mathcal{G}_{\mathcal{E}}\right)$$
%set of units $\left(  i,j\right)  $ in $\mathcal{G}_{\mathcal{E}}$ such that no
%vertex in $O_{i}$ is adjacent (neighbors) to a vertex in $O_{j}$. That is $j\not \in
%\mathcal{N}_{i}$ for all $i,j\in\mathcal{S}\left(  \mathcal{G}_{\mathcal{E}%}\right)  $. 
That is, a stable set is a set of nodes with the property that no two nodes in the set 
have an edge connecting them in the network. The size of a stable set is the number of units it contains. A
maximal stable set is a stable set such that no unit in $\mathcal{G}%
_{\mathcal{E}}$ can be added without violating the independence condition. A
maximum stable set $\mathcal{S}_{\max}\left(  \mathcal{G}_{\mathcal{E}%
}\right)  $ is a maximal stable set of largest possible size for
$\mathcal{G}_{\mathcal{E}}$. This size is called the stable number or
independence number of $\mathcal{G}_{\mathcal{E}}$, which we denote
$n_{1,N}=n_{1}\left(  \mathcal{G}_{\mathcal{E}}\right)  $. A maximum stable
set is not necessarily unique in a given graph, and finding one such set and
enumerating them all is challenging but a well-studied problem of computer
science. In fact, finding a maximum stable set is a well-known
NP-complete problem. Nevertheless, both exact and approximate algorithms exist
that are computationally more efficient than an exhaustive search. Exact
algorithms which identify all maximum stable sets were described in \cite{robson1986algorithms,makino2004new,fomin2009measure}.
Unfortunately, exact algorithms for finding maximum stable sets quickly become
computationally prohibitive with moderate to large networks. In fact, the
maximum stable set problem is known not to have an efficient approximation
algorithm unless P=NP \citep{zuckerman2006linear}. A practical approach we take in this
paper is to simply use an enumeration algorithm that lists a collection of
maximal stable sets \citep{myrvold2013fast}, and pick the largest of the
maximal sets found. \ Let $\Xi_{1}=\left\{  \mathcal{S}_{\max}\left(
\mathcal{G}_{\mathcal{E}}\right)  :\mathrm{card}\left(  \mathcal{S}_{\max
}\left(  \mathcal{G}_{\mathcal{E}}\right)  \right)  =\mathrm{n}_{1}\left(
\mathcal{G}_{\mathcal{E}}\right)  \right\}$ denote the collection of all
maximum (or largest identified maximal) stable sets for $\mathcal{G}_{\mathcal{E}}$. 

The Markov property associated with $\mathcal{G}_{\mathcal{E}}$ implies that
outcomes of units within such sets are mutually conditionally independent
given their Markov blankets. This implies the (partial) conditional likelihood
function which only involves units in the stable set factorizes, suggesting
that tools from maximum likelihood estimation may apply. In the Appendix, we
establish that this is in fact the case, in the sense that under certain
regularity conditions, coding maximum likelihood estimators of $\tau$ based on
maximum (or largest identified maximal) stable sets are consistent and asymptotically normal (CAN). Consider
the coding likelihood functions for $\tau_{Y}$ and $\tau_{L}$ based on a
stable set $\mathcal{S}_{\max}\left(  \mathcal{G}_{\mathcal{E}%
}\right)  \in$ $\Xi_{1}$:
\begin{align}
\mathcal{CL}_{Y}\left(  \tau_{Y}\right)   &  =%
%TCIMACRO{\dprod \limits_{i\in\mathcal{S}_{\max}\left(  \mathcal{G}%
%_{\mathcal{E}}\right)  }}%
%BeginExpansion
{\displaystyle\prod\limits_{i\in\mathcal{S}_{\max}\left(  \mathcal{G}%
		_{\mathcal{E}}\right)  }}
%EndExpansion
\mathcal{L}_{Y,\mathcal{S}_{\max}\left(  \mathcal{G}_{\mathcal{E}}\right)
	,i}\left(  \tau_{Y}\right)  =%
%TCIMACRO{\dprod \limits_{i\in\mathcal{S}_{\max}\left(  \mathcal{G}%
%_{\mathcal{E}}\right)  }}%
%BeginExpansion
{\displaystyle\prod\limits_{i\in\mathcal{S}_{\max}\left(  \mathcal{G}%
		_{\mathcal{E}}\right)  }}
%EndExpansion
f\left(  Y_{i}|\mathcal{O}_{i},A_{i},L_{i};\tau_{Y}\right)  ;\label{Res LIk}\\
\mathcal{CL}_{L}\left(  \tau_{L}\right)   &  =%
%TCIMACRO{\dprod \limits_{i\in\mathcal{S}_{\max}\left(  \mathcal{G}%
%_{\mathcal{E}}\right)  }}%
%BeginExpansion
{\displaystyle\prod\limits_{i\in\mathcal{S}_{\max}\left(  \mathcal{G}%
		_{\mathcal{E}}\right)  }}
%EndExpansion
\mathcal{L}_{L,\mathcal{S}_{\max}\left(  \mathcal{G}_{\mathcal{E}}\right)
	,i}\left(  \tau_{L}\right)  =%
%TCIMACRO{\dprod \limits_{i\in\mathcal{S}_{\max}\left(  \mathcal{G}%
%_{\mathcal{E}}\right)  }}%
%BeginExpansion
{\displaystyle\prod\limits_{i\in\mathcal{S}_{\max}\left(  \mathcal{G}%
		_{\mathcal{E}}\right)  }}
%EndExpansion
f\left(  L_{i}|\left\{  L_{j}:j\in\mathcal{N}_{i}\right\}  ;\tau_{L}\right)  .
\label{likL}%
\end{align}
The estimators $\widehat{\tau}_{Y}=\arg\max_{\tau_{Y}}\log\mathcal{CL}%
_{Y}\left(  \tau_{Y}\right)  $ and $\widehat{\tau}_{L}=\arg\max_{\tau_{L}}%
\log\mathcal{CL}_{L}\left(  \tau_{L}\right)  $ are analogous to Besag's coding
maximum likelihood estimators. Consider a network asymptotic theory according
to which $\{\mathcal{G}_{\mathcal{E}_{N}}: N\}$ is a sequence of chain graphs
as $N\rightarrow\infty,$ with vertices $(\mathbf{A}_{\mathcal{E}}%
\mathbf{,L}_{\mathcal{E}}\mathbf{,Y}_{\mathcal{E}}\mathbf{)}$ that follow
correctly specified auto-models with unknown parameters $\left(  \tau_{Y}%
,\tau_{L}\right)  $, and with edges defined according to a sequence of
networks $\mathcal{E}_{N}$, $N = 1, 2, \ldots$ of increasing size. We
establish the following result in the Appendix

\textit{Result 1: Suppose that}$~n_{1,N}\rightarrow\infty$ \textit{as}
$N\rightarrow\infty$ \textit{then under conditions 1-6 given in the
	Appendix,}
\begin{align*}
&  \widehat{\tau}_{L}\underset{N\longrightarrow\infty}{\longrightarrow}%
\tau\text{ \textit{in probability;}}\widehat{\tau}_{Y}%
\underset{N\longrightarrow\infty}{\longrightarrow}\tau\text{ \textit{in
		probability.} }\\
&  \sqrt{n_{1,N}}\Gamma_{n_{1,N}}^{1/2}\left(  \widehat{\tau}_{L}-\tau
_{L}\right)  \underset{N\longrightarrow\infty}{\longrightarrow}N\left(
0,I\right)  ;\\
&  \sqrt{n_{1,N}}\Omega_{n_{1,N}}^{1/2}\left(  \widehat{\tau}_{Y}-\tau
_{Y}\right)  \underset{N\longrightarrow\infty}{\longrightarrow}N\left(
0,I\right)  ;\\
\Gamma_{n_{1,N}}  &  =\frac{1}{n_{1,N}}\sum_{i\in\mathcal{S}_{\max}\left(
	\mathcal{G}_{\mathcal{E}_{N}}\right)  }\left\{  \frac{\partial\log
	\mathcal{CL}_{L,\mathcal{S}_{\max}\left(  \mathcal{G}_{\mathcal{E}_{N}%
		}\right)  ,i}\left(  \tau_{L}\right)  }{\partial\tau_{L}}\right\}  ^{\otimes
	2},\\
\Omega_{n_{1,N}}  &  =\frac{1}{n_{1,N}}\sum_{i\in\mathcal{S}_{\max}\left(
	\mathcal{G}_{\mathcal{E}_{N}}\right)  }\left\{  \frac{\partial\log
	\mathcal{CL}_{Y,\mathcal{S}_{\max}\left(  \mathcal{G}_{\mathcal{E}_{N}%
		}\right)  ,i}\left(  \tau_{Y}\right)  }{\partial\tau_{Y}}\right\}  ^{\otimes
	2}.
\end{align*}
\textit{ }

Note that by the information equality, $\Gamma_{n_{1,N}}$ and $\Omega
_{n_{1,N}}$ can be replaced by the standardized (by $n_{1,N}$) negative second
derivative matrix of corresponding coding log likelihood functions. Note also
that condition $n_{1,N}\rightarrow\infty$ as $N\rightarrow\infty$ essentially
rules out the presence of an ever-growing hub on the network as it expands
with $N$, thus ensuring that there is no small set of units in which majority
of connections are concentrated asymptotically. Suppose that each unit on a
network of size $N$ is connected to no more than $C_{\max}<N,$ then according
to Brooks' Theorem, the stable number $n_{1,N}$ satisfies the inequalities
\citep{brooks1941colouring}: 
\[
\frac{N}{C_{\max}+1}\leq n_{1,N}\leq N.
\]
This implies that in a network of bounded degree, $n_{1,N}=O\left(  N\right)
$ is guaranteed to be of the same order as the size of the network$;$ however
$n_{1,N}$ may grow at substantially slower rates $(n_{1,N}=o(N))$ if $C_{\max
}$ is unbounded.

\subsection*{3.5 Pseudo-likelihood estimation}

Note that because $L_{i}$ is likely multivariate, further computational
simplification can be achieved by replacing $f\left(  L_{i}|\left\{
L_{j}:j\in\mathcal{N}_{i}\right\}  ;\tau_{L}\right)  $ with the
pseudo-likelihood (PL) function
\[%
%TCIMACRO{\dprod \limits_{s=1}^{p}}%
%BeginExpansion
{\displaystyle\prod\limits_{s=1}^{p}}
%EndExpansion
f\left(  L_{s,i}|L_{\backslash s,i},\left\{  L_{j}:j\in\mathcal{N}%
_{i}\right\}  ;\tau_{L}\right)
\]
in equation $\left(  \ref{likL}\right)  .$ This substitution is
computationally more efficient as it obviates the need to evaluate a
multivariate integral in order to normalize the joint law of $L_{i}$. Let
$\widetilde{\tau}$ denote the estimator which maximizes the log of the
resulting modified coding likelihood function $\mathcal{L}_{L,\mathcal{S}%
	_{\max}\left(  \mathcal{G}_{\mathcal{E}}\right)  ,i}^{\ast}\left(  \tau
_{L}\right)  .$ It is straightforward using the proof of Result 1 to establish
that its covariance may be approximated by the sandwich formula $\Phi
_{L,n_{1,N}}^{-1}\Gamma_{L,n_{1,N}}\Phi_{L,n_{1,N}}^{-1}$ \citep{guyon1995random}, where
\begin{align*}
\Gamma_{L,n_{1,N}}  &  =\frac{1}{n_{1,N}}\sum_{i\in\mathcal{S}_{\max}\left(
	\mathcal{G}_{\mathcal{E}}\right)  }\left\{  \frac{\partial\log\mathcal{L}%
	_{L,\mathcal{S}_{\max}\left(  \mathcal{G}_{\mathcal{E}}\right)  ,i}^{\ast
	}\left(  \tau_{L}\right)  }{\partial\tau_{L}}\right\}  ^{\otimes2},\\
\Phi_{L,n_{1,N}}  &  =\frac{1}{n_{1,N}}\sum_{i\in\mathcal{S}_{\max}\left(
	\mathcal{G}_{\mathcal{E}}\right)  }\left\{  \frac{\partial^{2}\log
	\mathcal{L}_{L,\mathcal{S}_{\max}\left(  \mathcal{G}_{\mathcal{E}}\right)
		,i}^{\ast}\left(  \tau_{L}\right)  }{\partial\tau_{L}\partial\tau_{L}^{T}%
}\right\}  .
\end{align*}

As later illustrated in extensive simulation studies, coding estimators can be
inefficient, since the partial conditional likelihood function associated with
coding estimators disregards contributions of units $i\not \in \mathcal{S}%
_{\max}\left(  \mathcal{G}_{\mathcal{E}}\right)  .$ Substantial information
may be recovered by combining multiple coding estimators each obtained from a
separate approximate maximum stable set, however accounting for dependence
between the different estimators can be challenging.

Pseudo-likelihood (PL) estimation offers a simple alternative approach which
is potentially more efficient than either approach described above. PL
estimators maximize the log-PLs
\begin{align}
\log\left\{  \mathcal{PL}_{Y}\left(  \tau_{Y}\right)  \right\}   &  =%
%TCIMACRO{\dsum \limits_{i\in\mathcal{G}_{\mathcal{E}}}}%
%BeginExpansion
{\displaystyle\sum\limits_{i\in\mathcal{G}_{\mathcal{E}}}}
%EndExpansion
\log f\left(  Y_{i}|\mathcal{O}_{i},A_{i},L_{i};\tau_{Y}\right)  ;\\
\log\left\{  \mathcal{PL}_{L}\left(  \tau_{L}\right)  \right\}   &  =%
%TCIMACRO{\dsum \limits_{i\in\mathcal{G}_{\mathcal{E}}}}%
%BeginExpansion
{\displaystyle\sum\limits_{i\in\mathcal{G}_{\mathcal{E}}}}
%EndExpansion
\log f\left(  L_{k,i}| \left\{  L_{s,i}:s \in\{1,..,p\} \setminus k \right\} ,
\left\{  L_{s,j}: s \in\{1,..,p\}, j\in\mathcal{N}_{i}\right\}  ;\tau
_{L}\right)  .
\end{align}
Denote corresponding estimators $\check{\tau}_{Y}$ and $\check{\tau}_{L}$,
$\ $which are shown to be consistent in the Appendix. There however is
generally no guarantee that their asymptotic distribution follows a Gaussian
distribution due to complex dependence between units on the network
prohibiting application of the central limit theorem. As a consequence, for
inference, we recommend using the parametric bootstrap, whereby algorithm
Gibbs sampler I\ of Section 2.2 may be used to generate multiple bootstrap
samples from the observed data likelihood evaluated at $\left(  \check{\tau
}_{Y},\check{\tau}_{L}\right)  ,$ which in turn can be used to obtain a
bootstrap distribution for $\left(  \check{\tau}_{Y},\check{\tau}_{L}\right)
$ and corresponding inferences such as bootstrap quantile confidence intervals.

\section*{\centering 4. AUTO-G-COMPUTATION}

We now return to the main goal of the paper, which is to obtain valid
inferences about $\beta_{i}\left(  \mathbf{a}\right)  .$ The
auto-G-computation algorithm entails evaluating
\[
\widehat{\beta}_{i}\left(  \mathbf{a}\right)  \approx K^{-1}\sum_{k=0}%
^{K}\widehat{Y}_{i}^{(m^{\ast}+k)},
\]
where $\widehat{Y}_{i}^{(m)}$ are generated by Gibbs Sampler I algorithm under
posited auto-models with estimated parameters $\left(  \hat{\tau}_{Y}%
,\hat{\tau}_{L}\right)  .$ An analogous estimator $\breve{\beta}_{i}\left(
\mathbf{a}\right)  $ can be obtained using $\left(  \check{\tau}_{Y}%
,\check{\tau}_{L}\right)  $ instead of $\left(  \hat{\tau}_{Y},\hat{\tau}%
_{L}\right)  .$ In either case, the parametric bootstrap may be used in
conjunction with Gibbs Sampler I in order to generate the corresponding
bootstrap distribution of estimators of $\beta_{i}\left(  \mathbf{a}\right)  $
conditional on either $\widehat{\beta}_{i}\left(  \mathbf{a}\right)  $ or
$\breve{\beta}_{i}\left(  \mathbf{a}\right)  $. Alternatively, a less
computationally intensive approach first generates i.i.d. samples $\tau
_{Y}^{(j)}$ and $\tau_{L}^{(j)}$ $,$ $j=1,...J$ from $N\left(  \hat{\tau}%
_{Y},\widehat{\Gamma}_{n_{1,N}}\right)  $ and $N\left(  \hat{\tau}%
_{L},\widehat{\Omega}_{n_{1,N}}\right)  $ respectively, conditional on the
observed data, where $\widehat{\Gamma}_{n_{1,N}}=\Gamma_{n_{1,N}}\left(
\widehat{\tau}_{L}\right)  $ and $\widehat{\Omega}_{n_{1,N}}=\Omega_{n_{1,N}%
}\left(  \widehat{\tau}_{Y}\right)  $ estimate $\Gamma_{n_{1,N}}$ and
$\Omega_{n_{1,N}}.$. Next, one computes corresponding estimators
$\widehat{\beta}_{i}^{(j)}\left(  \mathbf{a}\right)  $ based on simulated data
generated using Gibbs Sampler I algorithm under $\tau_{Y}^{(j)}$ and $\tau
_{L}^{(j)},$ $j=1,...,J.$ The empirical distribution of $\left\{
\widehat{\beta}_{i}^{(j)}\left(  \mathbf{a}\right)  :j\right\}  $ may be used
to obtain standard errors for $\widehat{\beta}_{i}\left(  \mathbf{a}\right)
$, and corresponding Wald type or quantile-based confidence intervals for
direct and spillover causal effects.

\section*{\centering 5. SIMULATION STUDY}

We performed an extensive simulation study to evaluate the performance of the
proposed methods on networks of varying density and size. Specifically, we
investigated the properties of the coding-type and pseudo-likelihood
estimators of unknown parameters $\tau_{Y}$ and $\tau_{L}$ indexing the joint
observed data likelihood. Additionally, we evaluated the performance of
proposed estimators of the network counterfactual mean $\beta(\alpha
)=N^{-1}\sum_{i=1}^{N}\sum_{\mathbf{a}_{-i}\in\mathcal{A(}N)}\pi_{i}\left(
\mathbf{a}_{-i};\alpha\right)  E\left(  Y_{i}\left(  \mathbf{a}\right)
\right)  ,$ as well as for the direct effect $DE(\alpha),$ and the spillover
effect $IE(\alpha)$, where $\alpha$ is a specified treatment allocation law
described below.

We simulated three networks of size 800 with varying densities: low (each node
has either 2, 3, or 4 neighbors), medium (each node has either 5, 6, or 7
neighbors), and high (each node has either 8, 9, or 10 neighbors). For
reference, a depiction of the low density network of size 800 is given in
Figure 2. Additionally, we simulated low density networks of size 200, 400, and 1,000. 
The network graphs were all simulated in Wolfram Mathematica 10
using the RandomGraph function. For each network, we obtained
an (approximate) maximum stable set. The stable sets for the 800 node networks
were of size $n_{1,low}=375$, $n_{1,med}=275$, $n_{1,high}=224$.

\begin{figure}[tbh]
	\caption{Network of size 800 with low density }%
	\centering
	\includegraphics[width = 4in]{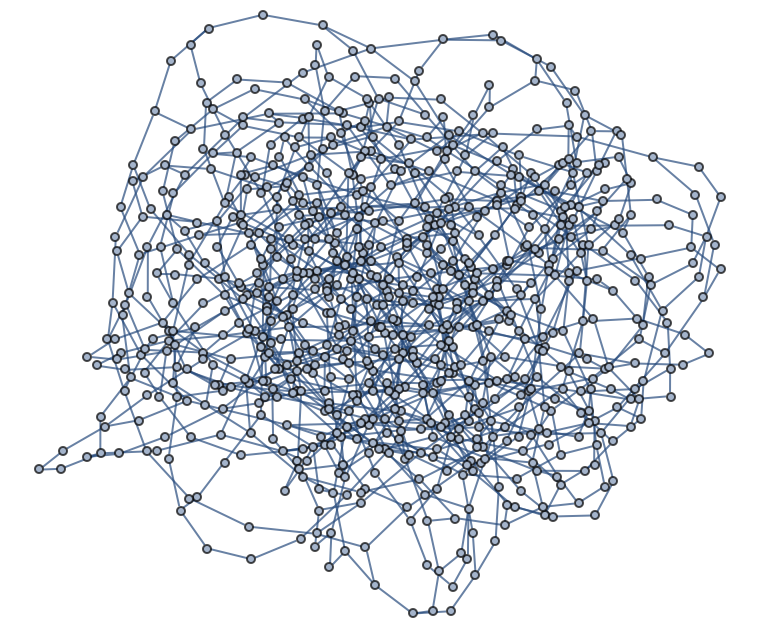}\end{figure}

\ For units $i=1,...,N$, we generated using Gibbs Sampler I$\ $a vector of
binary confounders $\{L_{1i},L_{2i},L_{3i}\}$, a binary treatment assignment
$A_{i}$, and a binary outcome $Y_{i}$ from the following auto-models
consistent with the chain graph induced by the simulated network:%

\footnotesize
\begin{align*}
Pr(L_{1,i}=1\mid\mathbf{L}_{\setminus1,i},\{\mathbf{L}_{1,j}:j\in
\mathcal{N}_{i}\})  &  =\mathrm{expit}\bigg(\tau_{1}+\rho_{12}L_{2,i}%
+\rho_{13}L_{3,i}+\nu_{11}\sum_{j\in\mathcal{N}_{i}}L_{1,j}+\nu_{12}\sum
_{j\in\mathcal{N}_{i}}L_{2,j}+\nu_{13}\sum_{j\in\mathcal{N}_{i}}%
L_{3,j}\bigg)\\
Pr(L_{2,i}=1\mid\mathbf{L}_{\setminus2,i},\{\mathbf{L}_{j}:j\in\mathcal{N}%
_{i}\})  &  =\mathrm{expit}\bigg(\tau_{2}+\rho_{12}L_{1,i}+\rho_{23}%
L_{3,i}+\nu_{21}\sum_{j\in\mathcal{N}_{i}}L_{1,j}+\nu_{22}\sum_{j\in
	\mathcal{N}_{i}}L_{2,j}+\nu_{23}\sum_{j\in\mathcal{N}_{i}}L_{3,j}\bigg)\\
Pr(L_{3,i}=1\mid\mid\mathbf{L}_{\setminus3,i},\{\mathbf{L}_{j}:j\in
\mathcal{N}_{i}\})  &  =\mathrm{expit}\bigg(\tau_{3}+\rho_{13}L_{1,i}%
+\rho_{23}L_{2,i}+\nu_{31}\sum_{j\in\mathcal{N}_{i}}L_{1,j}+\nu_{32}\sum
_{j\in\mathcal{N}_{i}}L_{2,j}+\nu_{33}\sum_{j\in\mathcal{N}_{i}}L_{3,j}\bigg)
\end{align*}

\begin{align*}
Pr(A_{i}=1\mid L_{i},\{A_{j},\mathbf{L}_{j}:j\in\mathcal{N}_{i}\})  &
=\mathrm{expit}\bigg(\gamma_{0}+\gamma_{1}L_{1,i}+\gamma_{2}\sum
_{j\in\mathcal{N}_{i}}L_{1,j}+\gamma_{3}L_{2,i}\\
&  \hspace{2cm}+\gamma_{4}\sum_{j\in\mathcal{N}_{i}}L_{2,j}+\gamma_{5}%
L_{3,i}+\gamma_{6}\sum_{j\in\mathcal{N}_{i}}L_{3,j}+\gamma_{7}\sum
_{j\in\mathcal{N}_{i}}A_{j}\bigg)
\end{align*}

\begin{align*}
Pr(Y_{i}=1\mid A_{i},L_{i},\mathcal{O}_{i})  &  =\mathrm{expit}\bigg(\beta
_{0}+\beta_{1}A_{i}+\beta_{2}\sum_{j\in\mathcal{N}_{i}}A_{j}+\beta_{3}%
L_{1,i}+\beta_{4}\sum_{j\in\mathcal{N}_{i}}L_{1,j}\\
&  \hspace{2cm}+\beta_{5}L_{2,i}+\beta_{6}\sum_{j\in\mathcal{N}_{i}}%
L_{2,j}+\beta_{7}L_{3,i}+\beta_{8}\sum_{j\in\mathcal{N}_{i}}L_{3,j}+\beta
_{9}\sum_{j\in\mathcal{N}_{i}}Y_{j}\bigg)
\end{align*}
\normalsize

\noindent where $\mathrm{expit}\left(  x\right)  =(1+\exp\left(  -x\right)
)^{-1},\tau_{L}=\{\tau_{1},\tau_{2},\tau_{3},\rho_{12},\rho_{13},\rho_{23}%
,\nu_{11},\nu_{12},\nu_{13},\nu_{22},\nu_{21},\nu_{23},\nu_{33},\nu_{31}%
,\nu_{32}\}$, $\tau_{A}=\{\gamma_{0},...,\gamma_{7}\}$, and $\tau_{Y}%
=\{\beta_{0},...,\beta_{9}\}$. We evaluated network average direct and
spillover effects via the Gibbs Sampler I algorithm under true parameter
values $\tau_{Y}$ and $\tau_{L}$ and a treatment allocation, $\alpha$ given by
a binomial distribution with event probability equal to $0.7$. All parameter
values are summarized in Table 1. \ We generated $S=1,000$ simulations of the
chain graph for each of the 4 simulated network structures. For each
simulation $s$, data were generated by running the Gibbs sampler I algorithm
$4,000$ times with the first $1,000$ iterations as burn-in. Additionally, we
thinned the chain by retaining every third realization to reduce
autocorrelation. \newline

\begin{table}[tbh]
	\caption{True parameter values for simulation study}
	\centering
	\begin{tabular}
		[c]{cc}\hline
		\textbf{Parameter} & \textbf{Truth}\\\hline
		$\tau_{L}$ & (-1.0,0.50,-0.50,0.1,0.2,0.1,0.1,0,0,0.1,0,0,0.1,0,0)\\
		$\tau_{A}$ & (-1.00,0.50,0.10,0.20,0.05,0.25,-0.08,0.30\\
		$\tau_{Y}$ &
		(-0.30,-0.60,-0.20,-0.20,-0.05,-0.10,-0.01,0.40,0.01,0.20)\\\hline
	\end{tabular}
\end{table}

\noindent For each realization of the chain graph, $\mathcal{G}_{\mathcal{E}%
	_{N},s}$, we estimated $\tau_{Y}$ via coding-type maximum likelihood
estimation and $\tau_{L}$ via the modified coding estimator. Both sets of
parameters were also estimated via maximum pseudo-likelihood estimation. For
each estimator we computed corresponding causal effect estimators, their
standard errors and $95\%$ Wald confidence intervals as outlined in previous
Sections. The estimation of the auto-model parameters was computed in 
R using functions \texttt{optim()} and \texttt{glm()}(R Core Team, 2013). 
The network average causal effects were estimated using Gibbs Sampler I 
using the \texttt{agcEffect} function in the  \texttt{autognet} R package by 
plugging in estimates for $(\tau_{L},\tau_{Y})$ using $K=50$ iterations and 
a burn-in of $m^{\ast}=10$ iterations. For variance estimation of the coding-type estimator, 
200 bootstrap replications were used.

Simulation results for the various density networks of size 800 are summarized
in Tables 2 and 3 for the following parameters: the network average
counterfactual $\beta(\alpha),$ the network average direct effect, and the
network average spillover effect. Both coding and pseudo-likelihood estimators
had small bias in estimating $\beta(\alpha)$ regardless of network density
(absolute bias $<0.01$). Coverage of the coding estimator ranged between
$93.1\%$ and $94.5\%$. Biases were also small for both spillover and direct
effects: the bias slightly increased with network density, but still stayed
below an absolute bias of $0.01$. Coverage of coding-based confidence
intervals for direct effects ranged from $92.5\%$ to $95.6\%$, while the
coverage for spillover effects decreased slightly with network density from
$93.7\%$ to $92.2\%$. It is important to note that as the network structure
changes with network size and density, the corresponding estimated parameters
likewise vary and therefore it is not necessarily straightforward to compare
performance of the methodology across network structure. Table 3 gives the MC
variance for the pseudo-likelihood estimator which confirms greater efficiency
compared to the coding estimator given the significantly larger effective
sample size used by pseudo-likelihood. Appendix Tables 1-3 report
bias and coverage for the network causal effect parameters for low density
networks of size 200, 400, and 1,000. Additionally, Appendix Figures 1 and 2
report bias and coverage for all 25 auto-model parameters in the
low-density network of size 800. As predicted by theory, coding-type and
pseudo-likelihood estimators exhibit small bias. Additionally, coding-type
estimators had approximately correct coverage, while pseudo-likelihood
estimators had coverage substantially lower than the nominal level for a
number of auto-model parameters. These results confirm the anticipated failure
of pseudo-likelihood estimators to be asymptotically Gaussian. Most notably,
the coverage for the outcome auto-model coefficient capturing dependence on
neighbors' outcomes $\beta_{9}$ was $81\%$, while coverage of the coding-type
Wald CI for this coefficient was $94\%$. Although not shown here, the coverage
results for the auto-model parameters are consistent across all simulations.

%NEWLY ADDED
We also assessed the performance of auto-g-computation in small, dense networks and in the presence of missing network edges. For the first, we generated one network of size 100  ($n_{1,100} = 25$) and an additional network of size 200 ($n_{1,200} = 57$). For the network of size 100, coding estimation of auto-model parameters in 437 of the 1,000 simulated samples had convergence issues due to the small size of the maximal independent set. Excluding results with convergence issues, the causal estimates were biased and did not have correct coverage (see Appendix Figure 3a). The performance for the network of size 200 was much improved across these endpoints, though oftentimes the confidence intervals were too wide to be informative. In both cases, the pseudo-likelihood estimator exhibited less bias than the coding estimator. In the previously described dense network of size 800, we randomly removed 564 (14\%) of edges. The estimated parameters from the auto-models were unbiased and had correct coverage (see Appendix Figure 4). However, the causal estimates for both the coding and pseudo-likelihood estimators exhibited bias, and the coding estimator had coverage slightly below the nominal level with the estimated spillover effect shifted towards null (see Appendix Figure 5). 

\begin{table}[!htbp]
	\caption{Simulation results of coding based estimators of network causal
		effects for networks of size 800 by density}%
	\centering
	\begin{tabular}
		[c]{cccccc}
		&  &  &  &  & \\[-1.8ex]\hline
		& Truth & Bias & MC Variance & Robust Variance & 95\% CI
		Coverage\\\hline
		\multicolumn{1}{l}{Low ($n_{1} = 375$)} &  &  &  &  & \\
		$\beta(\alpha)$ & 0.211 & $0.001$ & 0.001 & 0.001 & 0.945\\
		Spillover & -0.166 & $0.002$ & 0.004 & 0.004 & 0.937\\
		Direct & -0.179 & $0.002$ & 0.002 & 0.002 & 0.943\\
		&  &  &  &  & \\
		\multicolumn{1}{l}{Medium ($n_{1} = 275$)} &  &  &  &  & \\
		$\beta(\alpha)$ & 0.209 & $0.003$ & 0.001 & 0.002 & 0.931\\
		Spillover & -0.170 & $0.007$ & 0.013 & 0.015 & 0.925\\
		Direct & -0.178 & $<0.001$ & 0.003 & 0.003 & 0.925\\
		&  &  &  &  & \\
		\multicolumn{1}{l}{High ($n_{1} = 224$)} &  &  &  &  & \\
		$\beta(\alpha)$ & 0.208 & 0.004 & 0.005 & 0.004 & 0.937\\
		Spillover & -0.171 & 0.001 & 0.032 & 0.027 & 0.922 \\
		Direct & -0.177 & -0.001 & 0.004 & 0.004 & 0.956\\\hline
	\end{tabular}
\end{table}

\begin{table}[!htbp]
	\caption{Simulation results of pseudo-likelihood based estimators of network
		causal effects for networks of size 800 by density}%
	\centering
	\begin{tabular}
		[c]{cccc}
		&  &  & \\[-1.8ex]\hline
		& Truth & Absolute Bias & MC Variance\\\hline
		\multicolumn{1}{l}{Low} &  &  & \\
		$\beta(\alpha)$ & 0.211 & $0.001$ & $<0.001$\\
		Spillover & -0.166 & $0.002$ & 0.002\\
		Direct & -0.179 & $0.002$ & 0.001\\
		&  &  & \\
		\multicolumn{1}{l}{Medium} &  &  & \\
		$\beta(\alpha)$ & 0.209 & $0.003$ & 0.001\\
		Spillover & -0.170 & $0.007$ & 0.005\\
		Direct & -0.178 & $<0.001$ & 0.001\\
		&  &  & \\
		\multicolumn{1}{l}{High} &  &  & \\
		$\beta(\alpha)$ & 0.208 & $0.004$ & 0.001\\
		Spillover & -0.171 & $0.001$ & 0.006\\
		Direct & -0.177 & $-0.001$ & 0.001\\\hline
	\end{tabular}
\end{table}

%NEWLY ADDED
\section*{\centering 6. DATA APPLICATION}

We consider an application of the \textit{auto-g-computation} algorithm to the Networks, Norms, and HIV/STI Risk Among Youth (NNAHRAY) study to assess the effect of past incarceration on infection with HIV, STI, or Hepatitis C accounting for the network structure \citep{khan2009incarceration}. The NNAHRAY study was conducted in a New York neighborhood with epidemic HIV and widespread drug use from 2002-2005 \citep{friedman2008relative}. Through in-person interviews, information was collected regarding the respondents' demographic characteristics, incarceration history, sexual partnerships and histories, and past drug use. At the time of the interviews, respondents were also tested for HIV, gonorrhea, chlamydia, Herpes Simple Virus (HSV) 2, Hepatitis C virus (HCV), and syphilis. The study population we consider includes all interviewed persons with recorded results from their HIV, STI, and HCV tests ($n=8$ persons missing) for a total sample size of $N=457$ persons. We assume that HIV/STI/HCV status is missing completely at random. We defined a network tie (i.e. edge) as a sexual and/or injection drug use partnership in the past three months if at least one of the partners reported the relationship. The network structure is given in Figure 3. The number of partners (i.e. neighbors) for each respondent varied from none to 10 resulting in a maximal independent set of $n_1 = 274$.

We estimated the network-level spillover and direct effect of past incarceration on infection with HIV, STI, or Hepatitis C (HCV) under a Bernoulli allocation strategy with treatment probability equal to 0.50. Past incarceration was defined as any amount of jail time in the respondents' history. We accounted for confounding by Latino/a ethnicity, age, education, and past illicit drug use. The same models and estimation procedure detailed in the simulation section were utilized; note that $\nu_{ij}$ where $i \neq j$ were assumed to be 0. For comparison, the auto-model parameters were estimated using the coding-type and pseudolikelihood estimators. Network average spillover and direct effects were restricted to persons with at least one network tie. 

%auto-model results and causal estimate results
Table 4 gives the outcome auto-model parameter point estimates for the coding and pseudolikelihood estimators with 95\% confidence intervals for the coding estimators excluding the covariate terms. Due to scaling by number of network ties, the outcome and exposure influence of network ties can be interpreted as the effect of average covariate value among network ties. Individuals who experienced prior incarceration had 2.12 [95\% CI: 1.07-4.21] times the odds of infection with HIV/STI/HCV compared to those without prior incarceration. However, the incarceration status of network ties was not significantly associated with a person's risk of HIV/STI/HCV (OR= 1.21 [95\% CI: 0.52-2.84]) conditional on the neighbors' outcomes. Individual's with a greater proportion of their ties infected with HIV, STI, and/or HCV were much more likely to be infected with HIV, STI, and/or HCV (OR=3.07 [95\% CI: 1.33-7.09]). The pseudolikelihood point estimates were similar to the coding results. The full results for auto-model parameters from both the covariate and outcome model are given in Appendix Figure 6. The network average direct effect is 0.14 [95\% CI: 0.02-0.28] when the proportion of persons with prior history of incarceration is 0.50. There was no significant evidence of a spillover effect of incarceration on HIV/STI/HCV risk over the network, as increasing the proportion of persons with a history of incarceration from 0 to 0.50 resulted in a negligible increase in average HIV/STI/HCV risk of a person with no prior incarceration [$\widehat{DE}$=0.04; 95\% CI: -0.06-0.14]. 

In the Appendix, we have included two alternate outcome auto-model specifications that incorporate the number of sexual and injection drug use partners for each person in the network. In an infection disease setting, the number of partners should in principle be accounted for in the analysis as it is likely a confounder for the effect of incarceration (both individual and neighbors' status) on infection status \citep{khan2018dissolution}. As shown in the Appendix, adjusting for the number of network ties (e.g  sexual and injection drug partners) did not change our conclusions. Following a reviewer's recommendation, we performed a simulation study based on the NNAHRAY network under under the sharp null and verified that we have valid inference in this setting. Results are provided in the Appendix Table 7.

\begin{figure}[!htbp]
	\centering
	\begin{minipage}{.8\textwidth}
		\centering
		\caption{Network graph from the NNAHRAY data (N=457) with individuals in the maximal independent set ($n_1 = 274$) in blue.}
		\includegraphics[scale=.70]{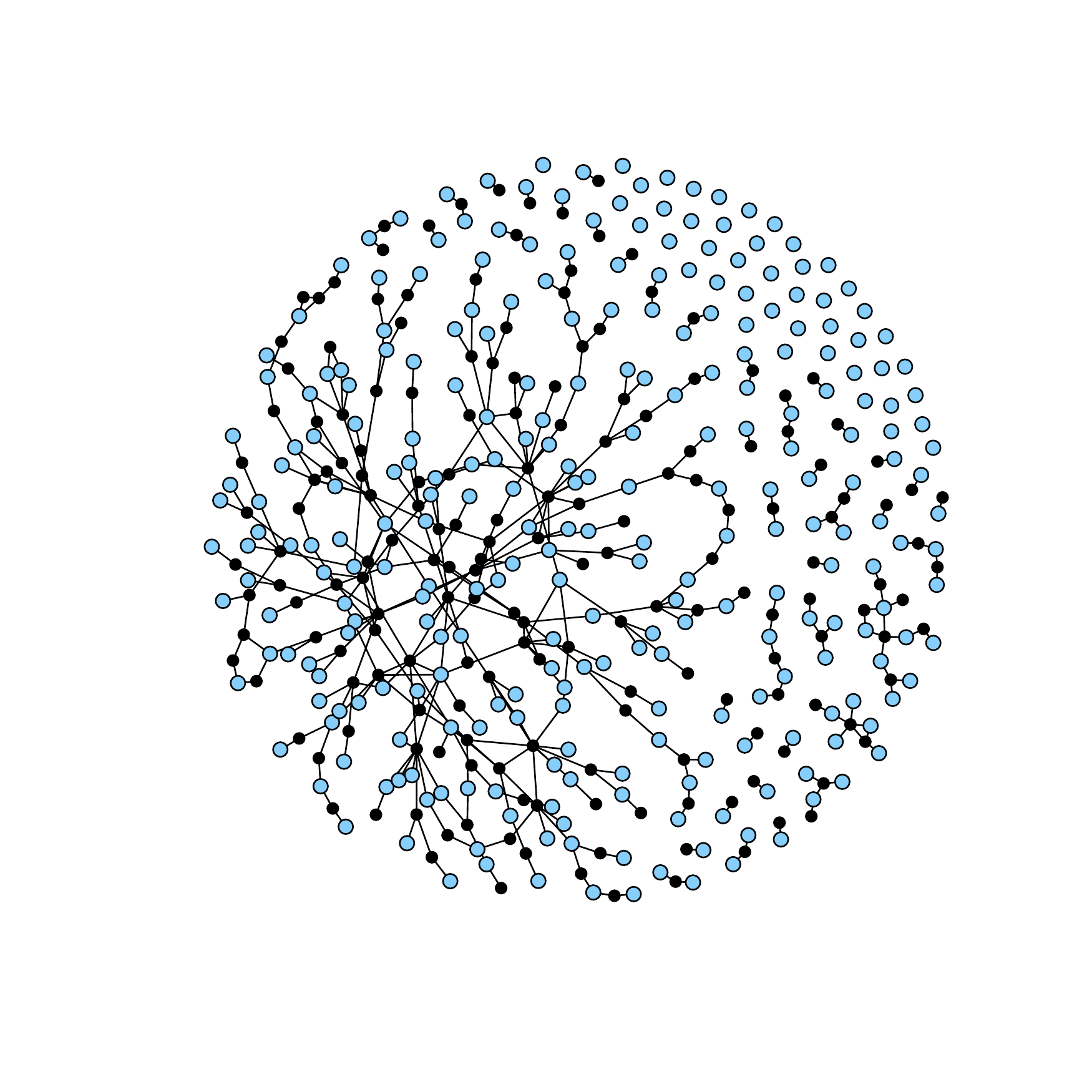} 
	\end{minipage}
\end{figure}

\begin{table}[!htbp]
	\caption{Outcome auto-model parameters estimates for coding and pseudolikelihood estimators (excluding covariates) on the odds ratio scale}%
	\centering
	\begin{tabular}
		[c]{lcc}
		&  & \\[-1.8ex]\hline
		& Estimates & 95\% CI \\\hline
		\multicolumn{1}{l}{\textbf{Coding}} &  &  \\
		\ \ Past incarceration status (individual) & 2.12 & [1.07, 4.21] \\
		\ \ Past incarceration status (neighbors)  & 1.21 & [0.52, 2.84]  \\
		\ \ HIV/STI/HCV status (neighbors) & 3.07 & [1.33, 7.09]  \\
		&  &  \\
		\multicolumn{1}{l}{\textbf{Pseudolikelihood}}& &   \\
		\ \ Past incarceration status (individual) & 2.36 & -- \\
		\ \ Past incarceration status (neighbors)  & 0.97 & --\\
		\ \ HIV/STI/HCV status (neighbors) & 2.62 & -- \\\hline
	\end{tabular}
\end{table}

\section*{\centering 7. CONCLUSION}

We have described a new approach for evaluating causal effects on a network of
connected units. Our methodology relies on the crucial assumption that
accurate information on network ties between observed units is available to
the analyst, which may not always be the case in practice. In fact, as demonstrated in our simulation study, bias may ensue if information about the network is incomplete, and therefore omits to account for all existing ties.\ In future work, we plan to further develop our methods to appropriately account for uncertainty about the underlying network structure.

Another limitation of the proposed approach is that it relies heavily on
parametric assumptions and as a result may be open to bias due to model
mis-specification. Although this limitation also applies to standard
g-computation for i.i.d settings which nevertheless has gained prominence in
epidemiology \citep{taubman2009intervening,robins2004effects,daniel2011gformula}, our
parametric auto-models which are inherently non-i.i.d may be substantially
more complex, as they must appropriately account both for outcome and
covariate dependence, as well as for interference. Developing appropriate
goodness-of-fit tests for auto-models is clearly a priority for future
research. In addition, to further alleviate concerns about modeling bias, we
plan in future work to extend semiparametric models such as structural nested
models to the network context.\ Such developments may offer a real opportunity
for more robust inference about network causal effects.

\section*{\centering SUPPLEMENTARY MATERIALS}

\begin{description}

\item[Acknowledgments:] We are grateful to Dr. Samuel R. Friedman at National Development and Research Institutes, Inc. for  access to the Networks, Norms, and HIV/STI Risk Among Youth study data and contributions to the data application section.

\item[Appendix:] Theorems, detailed proofs, and additional simulation results
(.zip file)

\item[Code:] Code for estimation and inference of network causal effects. To
download, please visit: https://isabelfulcher.github.io/autoGnetworks/ (R)
\end{description}

\newpage
\bibliography{references}

\newpage 

\begin{center}
{\LARGE \textit{Appendix}}
	
\end{center}

Throughout, we assume that we observe a vector of Random Fields $\left(
\mathbf{Y}_{N},\mathbf{A}_{N},\mathbf{L}_{N}\right)  ,$ such that
$\mathbf{Y}_{N}$ is a conditional Markov Random Field (MRF) given
$(\mathbf{A}_{N},\mathbf{L}_{N})$ on the sequence of Chain Graphs (CG)
$\mathcal{G}_{\mathcal{E}}$ associated with an increasing sequence of networks
$\mathcal{E}_{N},$ with distribution uniquely specified by the parametric
model for its Gibbs factors $f\left(  Y_{i}|\partial_{i};\tau_{Y}\right)  $
where $\partial_{i}=\left(  \mathcal{O}_{i},A_{i},L_{i}\right)  $ for all
$i=1,...,N.$

Let $\mathcal{N}_{i}^{(k)}$ denote the $kth$ order neighborhood of unit $i,$
defined as followed: $\mathcal{N}_{i}^{(1)}=$ $\mathcal{N}_{i},$
$\mathcal{N}_{i}^{(2)}=%
%TCIMACRO{\dbigcup \limits_{j\in\mathcal{N}_{i}^{(1)}}}%
%BeginExpansion
{\displaystyle\bigcup\limits_{j\in\mathcal{N}_{i}^{(1)}}}
%EndExpansion
\mathcal{N}_{j}^{(1)}\backslash\left(  \mathcal{N}_{i}^{(1)}\cup\{i\}\right)
,...,\mathcal{N}_{i}^{(k)}=%
%TCIMACRO{\dbigcup \limits_{j\in\mathcal{N}_{i}^{(k-1)}}}%
%BeginExpansion
{\displaystyle\bigcup\limits_{j\in\mathcal{N}_{i}^{(k-1)}}}
%EndExpansion
\mathcal{N}_{j}^{(1)}\backslash\left(
%TCIMACRO{\dbigcup \limits_{s\leq k-1}}%
%BeginExpansion
{\displaystyle\bigcup\limits_{s\leq k-1}}
%EndExpansion
\mathcal{N}_{i}^{(s)}\cup\{i\}\right)  .$ A\emph{ k-stable set} $\mathcal{S}%
^{(k)}\left(  \mathcal{G}_{\mathcal{E}},k\right)  $ of $\mathcal{G}%
_{\mathcal{E}}$ is a set of units $\left(  i,j\right)  $ in $\mathcal{G}%
_{\mathcal{E}}$ such that $%
%TCIMACRO{\dbigcup \limits_{s\leq k}}%
%BeginExpansion
{\displaystyle\bigcup\limits_{s\leq k}}
%EndExpansion
\mathcal{N}_{i}^{(s)}$ and $%
%TCIMACRO{\dbigcup \limits_{s\leq k}}%
%BeginExpansion
{\displaystyle\bigcup\limits_{s\leq k}}
%EndExpansion
\mathcal{N}_{j}^{(s)}$ are not neighbors$.$ The size of a k-stable set is the
number of units it contains. A maximal k-stable set is a k-stable set such
that no unit in $\mathcal{G}_{\mathcal{E}}$ can be added without violating the
independence condition. \ A maximum k-stable set $\mathcal{S}_{k,\max}\left(
\mathcal{G}_{\mathcal{E}}\right)  $ is a maximal k-stable set of largest
possible size for $\mathcal{G}_{\mathcal{E}}$. This size is called the
k-stable number of $\mathcal{G}_{\mathcal{E}}$, which we denote $n_{k,N}%
=n_{k}\left(  \mathcal{G}_{\mathcal{E}}\right)  $. Let $\Xi_{k}=\left\{
\mathcal{S}_{k,\max}\left(  \mathcal{G}_{\mathcal{E}}\right)  :\mathrm{card}%
\left(  \mathcal{S}_{k,\max}\left(  \mathcal{G}_{\mathcal{E}}\right)  \right)
=\mathrm{s}_{k}\left(  \mathcal{G}_{\mathcal{E}}\right)  \right\}  $ denote
the collection of all (approximate) maximum k-stable sets for $\mathcal{G}%
_{\mathcal{E}}$.

The unknown parameter $\tau_{Y}$ is in the interior of a compact
$\Theta\subset\mathbb{R}^{p}.$Let $\partial\partial_{i}=\left\{
\mathcal{O}_{j},A_{j},L_{j}:j\in\mathcal{N}_{i}^{(2)}\right\}  \cup\left\{
A_{j},L_{j}:j\in\mathcal{N}_{i}^{(3)}\right\}  \cup\left\{  L_{j}%
:j\in\mathcal{N}_{i}^{4)}\right\}  .$ Let
\[
\mathcal{M}_{Y,i}\left(  \tau_{Y},t;\partial_{i}\right)  =-E\left\{  \log
\frac{f\left(  Y_{i}|\mathcal{O}_{i},A_{i},L_{i};t\right)  }{f\left(
	Y_{i}|\mathcal{O}_{i},A_{i},L_{i};\tau_{Y}\right)  }|\partial_{i}\right\}
\geq0;
\]
We make the following assumptions:

\underline{{\large Regularity conditions:}}

Suppose that

\begin{enumerate}
	\item $\mathcal{S}\left(  \mathcal{G}_{\mathcal{E}}\right)  $ is a 1-stable set.
	
	\item $\mathcal{S}\left(  \mathcal{G}_{\mathcal{E}}\right)  $ can be
	partitioned into $K$ 4-stable subsets $\left\{  \mathcal{S}^{k}\left(
	\mathcal{G}_{\mathcal{E}}\right)  :k=1,...,K\right\}  $ such that
	$\mathcal{S}\left(  \mathcal{G}_{\mathcal{E}}\right)  =%
	%TCIMACRO{\dbigcup \limits_{k=1}^{K}}%
	%BeginExpansion
	{\displaystyle\bigcup\limits_{k=1}^{K}}
	%EndExpansion
	\mathcal{S}^{k}\left(  \mathcal{G}_{\mathcal{E}}\right)  .$ Let $n^{(k)}$
	denote the number of units in $\mathcal{S}^{k}\left(  \mathcal{G}%
	_{\mathcal{E}}\right)  $ and $n=\sum_{k}n^{(k)}$ denote the total number of
	units in $\mathcal{S}\left(  \mathcal{G}_{\mathcal{E}}\right)  .$ We further
	assume that there is a $k_{0}$ such that:
	
	\begin{enumerate}
		\item $\lim\inf_{\mathcal{E}}$ $n^{(k_{0})}/n>0$ and $\lim\inf_{\mathcal{E}}$
		$n/N>0.$
		
		\item Let $\mathcal{D}_{i}$ denote the support $\partial_{i}=\left(
		\mathcal{O}_{i},A_{i},L_{i}\right)  .$ The joint support $\mathcal{D=}%
		%TCIMACRO{\dprod \limits_{j\in\mathcal{N}_{i}}}%
		%BeginExpansion
		{\displaystyle\prod\limits_{j\in\mathcal{N}_{i}}}
		%EndExpansion
		\mathcal{D}_{j}$ is a fixed space of neighborhood configurations when
		$i\in\mathcal{S}^{k_{0}}\left(  \mathcal{G}_{\mathcal{E}}\right)  .$
	\end{enumerate}
	
	\item $\exists c>0$ such that for all $i\in\mathcal{S}^{k_{0}}\left(
	\mathcal{G}_{\mathcal{E}}\right)  ,$ we have for all values y$_{i},$
	$\partial_{i},\partial\partial_{i}$ and $t\in\Theta,$%
	\begin{align*}
	f\left(  y_{i}|\partial_{i};t_{y}\right)   &  >c\\
	f\left(  \partial_{i}|\partial\partial_{i},t_{y}\right)   &  >c
	\end{align*}
	w.here $f\left(  y_{i}|\partial_{i};t_{y}\right)  $ and $f\left(  \partial
	_{i}|\partial\partial_{i},t_{y}\right)  $ are density functions where
	$f\left(  y_{i}|\partial_{i};t_{y}\right)  $ is uniformly (in $i,y_{i}%
	,\partial_{i})$ continuous in $t_{y.}$
	
	\item There exists $\mathcal{M}_{y}\left(  \tau_{Y},t;z\right)  \geq0$ with
	$\left(  t,z\right)  \in\Theta\times\mathcal{D}$, $\mu-$integrable for all $t$
	such that:
	
	\begin{enumerate}
		\item $\mathcal{M}_{Y,i}\left(  \tau_{Y},t;z\right)  \geq\mathcal{M}%
		_{Y}\left(  \tau_{Y},t;z\right)  $ if $i\in\mathcal{S}^{k_{0}}\left(
		\mathcal{G}_{\mathcal{E}}\right)  .$
		
		\item $t\rightarrow N\left(  \tau_{Y},t\right)  =\int_{\mathcal{D}}%
		\mathcal{M}_{Y}\left(  \tau_{Y},t;z\right)  \mu\left(  dz\right)  $ is
		continuous and has a unique minimum at $t=\tau_{Y}.$
	\end{enumerate}
	
	\item For all $i\in\mathcal{S}\left(  \mathcal{G}_{\mathcal{E}}\right)  ,$
	$f\left(  y_{i}|\partial_{i};t_{y}\right)  $ admits three continuous
	derivatives at $t_{y}$ in a neighborhood of $\tau_{Y}$, and $1/f\left(
	y_{i}|\partial_{i};t_{y}\right)  $ and the $vth$ derivatives $f^{(v)}\left(
	y_{i}|\partial_{i};t_{y}\right)  $ $v=1,2,3$ are uniformly bounded in
	$i,y_{i},\partial_{i}$ and $t_{y}$ in a neighborhood of $\tau_{Y}.$
	
	\item There exist a positive-definite symmetric non-random $p\times p$ matrix
	$I\left(  \tau_{Y}\right)  $ such that
	\[
	\lim\inf_{\mathcal{E}}\Omega_{n_{1,N}}\geq I\left(  \tau_{Y}\right)
	\]
	
\end{enumerate}

\begin{theorem}
	Under assumptions 1-4, the maximum coding likelihood estimator on
	$\mathcal{S}\left(  \mathcal{G}_{\mathcal{E}}\right)  $ is consistent as
	$N\rightarrow\infty,$ that is
	\[
	\hat{\tau}_{Y}\rightarrow\tau_{Y}\text{ in probability }%
	\]
	where
	\[
	\hat{\tau}_{Y}=\arg\max_{t}\sum_{i\in\mathcal{S}\left(  \mathcal{G}%
		_{\mathcal{E}}\right)  }\log f\left(  Y_{i}|\mathcal{O}_{i},A_{i}%
	,L_{i};t\right)
	\]
	
\end{theorem}

\begin{proof}
	Define
	\[
	U_{\mathcal{E}}^{k_{0}}\left(  t\right)  =-\frac{1}{n^{(k_{0})}}\sum
	_{i\in\mathcal{S}^{k_{0}}\left(  \mathcal{G}_{\mathcal{E}}\right)  }\log
	f\left(  Y_{i}|\partial_{i};t\right)
	\]
	and write
	\[
	Z_{\mathcal{E}}=-\frac{1}{n^{(k_{0})}}\sum_{i\in\mathcal{S}^{k_{0}}\left(
		\mathcal{G}_{\mathcal{E}}\right)  }\left\{  \log\frac{f\left(  Y_{i}%
		|\partial_{i};t\right)  }{f\left(  Y_{i}|\partial_{i};\tau_{Y}\right)
	}+\mathcal{M}_{Y,i}\left(  \tau_{Y},t;\partial_{i}\right)  \right\}  +\frac
	{1}{n^{(k_{0})}}\sum_{i\in\mathcal{S}^{k_{0}}\left(  \mathcal{G}_{\mathcal{E}%
		}\right)  }\mathcal{M}_{Y,i}\left(  \tau_{Y},t;\partial_{i}\right)
	\]
	As the term in curly braces is the sum of centered random variables with
	bounded variance and conditionally independent given $\left\{  \partial
	_{i}:i\in\mathcal{S}^{k_{0}}\left(  \mathcal{G}_{\mathcal{E}}\right)
	\right\}  ,$
	\[
	\lim_{N\rightarrow\infty}\frac{1}{n^{(k_{0})}}\sum_{i\in\mathcal{S}^{k_{0}%
		}\left(  \mathcal{G}_{\mathcal{E}}\right)  }\left\{  \log\frac{f\left(
		Y_{i}|\partial_{i};t\right)  }{f\left(  Y_{i}|\partial_{i};\tau_{Y}\right)
	}+\mathcal{M}_{Y,i}\left(  \tau_{Y},t;\partial_{i}\right)  \right\}  =0\text{
	}a.s.
	\]
	Then we deduce the following sequence of inequalities almost surely:
	\begin{align*}
	\lim\inf_{\mathcal{E}}Z_{\mathcal{E}}  &  =\lim\inf_{\mathcal{E}}\left\{
	\frac{1}{n^{(k_{0})}}\sum_{i\in\mathcal{S}^{k_{0}}\left(  \mathcal{G}%
		_{\mathcal{E}}\right)  }\mathcal{M}_{Y,i}\left(  \tau_{Y},t;\partial
	_{i}\right)  \right\} \\
	&  \geq\lim\inf_{\mathcal{E}}\int_{\mathcal{D}}\mathcal{M}_{Y}\left(  \tau
	_{Y},t;z\right)  F_{n}\left(  \mathcal{S}^{k_{0}}\left(  \mathcal{G}%
	_{\mathcal{E}}\right)  ,dz\right)
	\end{align*}
	where
	\[
	F_{n}\left(  \mathcal{S}^{k_{0}}\left(  \mathcal{G}_{\mathcal{E}}\right)
	,dz\right)  =\frac{1}{n^{(k_{0})}}\sum_{i\in\mathcal{S}^{k_{0}}\left(
		\mathcal{G}_{\mathcal{E}}\right)  }1\left(  \partial i\in dz\right)
	\]
	By assumptions 1-3 and Lemma 5.2.2 of Guyon (1995), then there is a positive
	constant $c^{\ast}$ such that
	\[
	\lim\inf_{\mathcal{E}}F_{n}\left(  \mathcal{S}^{k_{0}}\left(  \mathcal{G}%
	_{\mathcal{E}}\right)  ,dz\right)  \geq c^{\ast}\lambda\left(  dz\right)
	\]
	therefore
	\begin{align*}
	\lim\inf_{\mathcal{E}}Z_{\mathcal{E}}  &  \geq c^{\ast}\int_{\mathcal{D}%
	}\mathcal{M}_{Y}\left(  \tau_{Y},t;z\right)  \lambda\left(  dz\right) \\
	&  \equiv c^{\ast}N\left(  \tau_{Y},t\right)
	\end{align*}
	Then note that%
	
	\begin{align*}
	U_{\mathcal{E}}\left(  t\right)   &  =-\frac{1}{n}\sum_{i\in\mathcal{S}\left(
		\mathcal{G}_{\mathcal{E}}\right)  }\log f\left(  Y_{i}|\partial_{i};t\right)
	\\
	&  =\sum_{k=1}^{K}\frac{n^{(k)}}{n}U_{\mathcal{E}}^{k}\left(  t\right)
	\end{align*}
	where
	\[
	U_{\mathcal{E}}^{k}\left(  t\right)  =-\frac{1}{n^{(k)}}\sum_{i\in
		\mathcal{S}^{k}\left(  \mathcal{G}_{\mathcal{E}}\right)  }\log f\left(
	Y_{i}|\partial_{i};t\right)
	\]
	Consistency of the coding maximum likelihood estimator then follows by
	assumption 1-4 and Corollary 3.4.1 of Guyon (1995).\ Consistency of the pseudo
	maximum likelihood estimator follows from
	\begin{align*}
	U_{\mathcal{E}}^{P}\left(  t\right)   &  \equiv-\frac{1}{N}\sum_{i}\log
	f\left(  Y_{i}|\partial_{i};t\right) \\
	&  =\frac{n}{N}U_{\mathcal{E}}\left(  t\right)  +\frac{\left(  N-n\right)
	}{N}\overline{U}_{\mathcal{E}}\left(  t\right)
	\end{align*}
	$U_{\mathcal{E}}^{P}\left(  t\right)  \equiv\left(  N-n\right)  ^{-1}%
	\sum_{i\not \in \mathcal{S}\left(  \mathcal{G}_{\mathcal{E}}\right)  }\log
	f\left(  Y_{i}|\partial_{i};t\right)  $ and a further application of Corollary
	3.4.1 of Guyon (1995).
\end{proof}

The next result establishes asymptotic normality of $\hat{\tau}_{Y}.$

\begin{theorem}
	Under Assumptions 1-6, as $N\rightarrow\infty$\textit{,}
	\begin{align*}
	&  \sqrt{n}\Omega_{n}^{1/2}\left(  \widehat{\tau}_{Y}-\tau_{Y}\right)
	\underset{N\longrightarrow\infty}{\longrightarrow}N\left(  0,I\right)  ;\\
	\Omega_{n}  &  =\frac{1}{n}\sum_{i\in\mathcal{S}\left(  \mathcal{G}%
		_{\mathcal{E}}\right)  }\left\{  \frac{\partial\log\mathcal{CL}_{Y,\mathcal{S}%
			\left(  \mathcal{G}_{\mathcal{E}}\right)  ,i}\left(  \tau_{Y}\right)
	}{\partial\tau_{Y}}\right\}  ^{\otimes2}.
	\end{align*}
	
\end{theorem}

\begin{proof}
	$\hat{\tau}_{Y}$ solves $V_{n}\left(  \hat{\tau}_{Y}\right)  =\frac
	{\partial\log\mathcal{CL}_{Y}\left(  \tau_{Y}\right)  }{\partial\tau_{Y}%
	}|_{\hat{\tau}_{Y}}=0,$ therefore
	\begin{align*}
	0  &  =\sqrt{n}V_{n}\left(  \tau_{Y}\right)  +\dot{V}_{n}\left(  \tau_{Y}%
	,\hat{\tau}_{Y}\right)  \sqrt{n}\left(  \hat{\tau}_{Y}-\tau_{Y}\right) \\
	V_{n}\left(  \tau_{Y}\right)   &  =\frac{1}{\sqrt{n}}\sum_{i\in\mathcal{S}%
		\left(  \mathcal{G}_{\mathcal{E}}\right)  }\mathcal{D}_{i}=\frac{1}{\sqrt{n}%
	}\sum_{i\in\mathcal{S}\left(  \mathcal{G}_{\mathcal{E}}\right)  }\left\{
	\frac{\partial\log\mathcal{CL}_{Y,\mathcal{S}\left(  \mathcal{G}_{\mathcal{E}%
			}\right)  ,i}\left(  \tau_{Y}\right)  }{\partial\tau_{Y}}\right\} \\
	\dot{V}_{n}\left(  \tau_{Y},\hat{\tau}_{Y}\right)   &  =\int_{0}^{1}%
	\frac{\partial^{2}\log\mathcal{CL}_{Y}\left(  \tau_{Y}\right)  }{\partial
		\tau_{Y}\partial\tau_{Y}^{T}}|_{t\left(  \hat{\tau}_{Y}-\tau_{Y}\right)
		+\tau_{Y}}dt
	\end{align*}
	As the variables $\mathcal{D}_{i}$ are centered, bounded and independent
	conditionally on $\left\{  \partial_{i}:i\in\mathcal{S}^{k_{0}}\left(
	\mathcal{G}_{\mathcal{E}}\right)  \right\}  ,$ one may apply a central limit
	theorem for non-iid bounded variable (Breiman, 1992). Under assumptions 5 and
	6, it follows that $\dot{V}_{n}\left(  \tau_{Y},\hat{\tau}_{Y}\right)
	+\Omega_{n}\left(  \tau_{Y}\right)  \rightarrow^{P}0_{p\times p},$ proving the result.
\end{proof}

Proofs of consistency of coding and pseudo maximum likelihood estimators of
$\tau_{L}$ , as well as asymptotic normality of coding estimator of $\tau_{L}$
follow along the same lines as above, upon substituting $\partial_{i}=\left\{
L_{j}:j\in\mathcal{N}_{i}\right\}  ,$ and replacing Assumption 2 with the
assumption that $\mathcal{S}\left(  \mathcal{G}_{\mathcal{E}}\right)  $ can be
partitioned into $K$ 2-stable subsets $\left\{  \mathcal{S}^{k}\left(
\mathcal{G}_{\mathcal{E}}\right)  :k=1,...,K\right\}  $ such that
$\mathcal{S}\left(  \mathcal{G}_{\mathcal{E}}\right)  =%
%TCIMACRO{\dbigcup \limits_{k=1}^{K}}%
%BeginExpansion
{\displaystyle\bigcup\limits_{k=1}^{K}}
%EndExpansion
\mathcal{S}^{k}\left(  \mathcal{G}_{\mathcal{E}}\right)  ,$ and that there is
a $k_{0}$ such that Assumptions 2.a. and 2.b. are satisfied.

\pagebreak
\section*{Additional simulation results: varying sample size} 

\begin{figure}[!htbp] 
	\caption{Simulation results of coding and pseudo-likelihood based estimators of covariate model gibbs factors for low density network of size 800 ($n_1 = 375$)}
	\includegraphics[width=\linewidth]{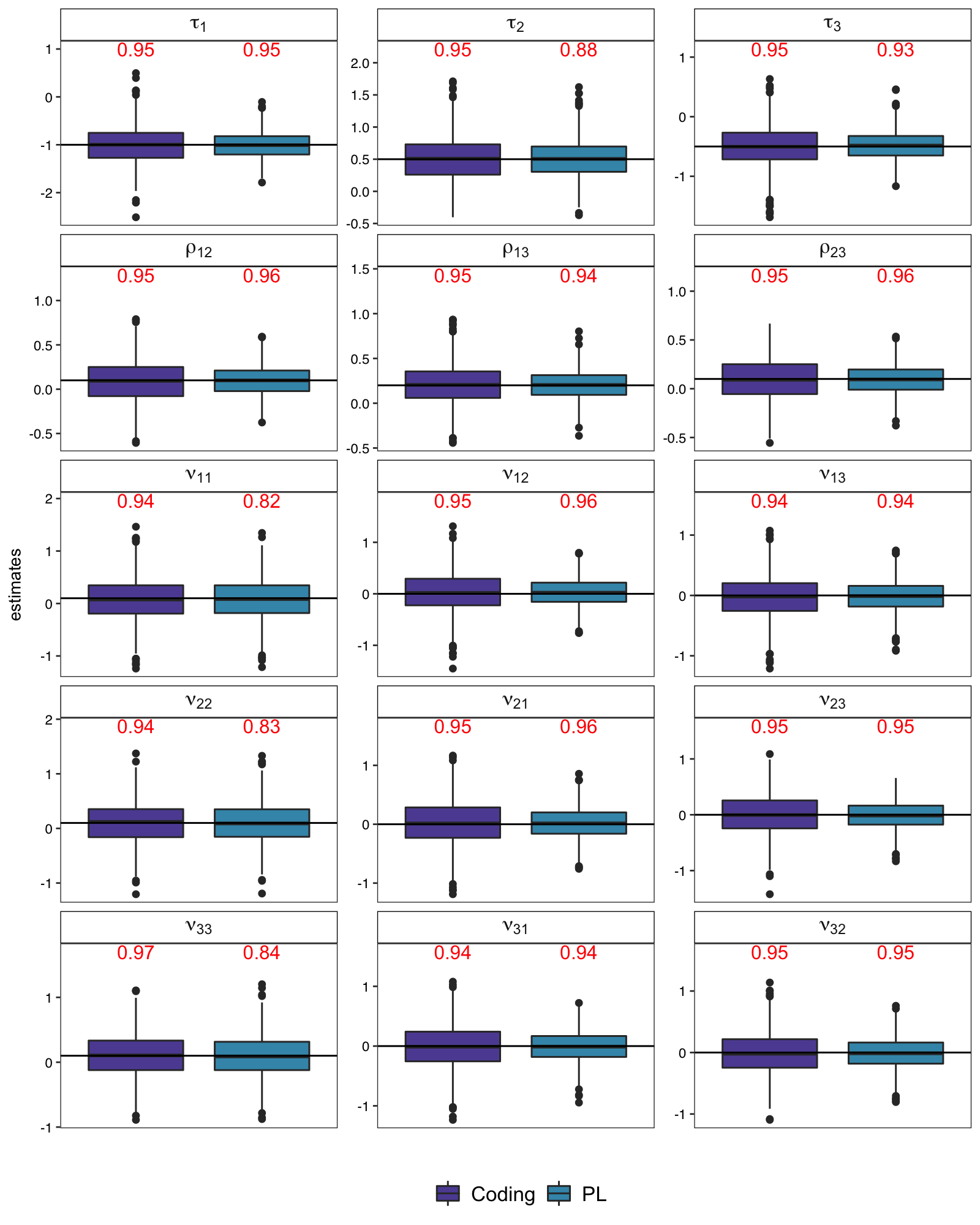}
\end{figure}

\begin{figure}[!htbp] 
	\caption{Simulation results of coding and pseudo-likelihood based estimators of outcome model gibbs factors for low density network of size 800 ($n_1 = 375$)}
	\includegraphics[width=\linewidth]{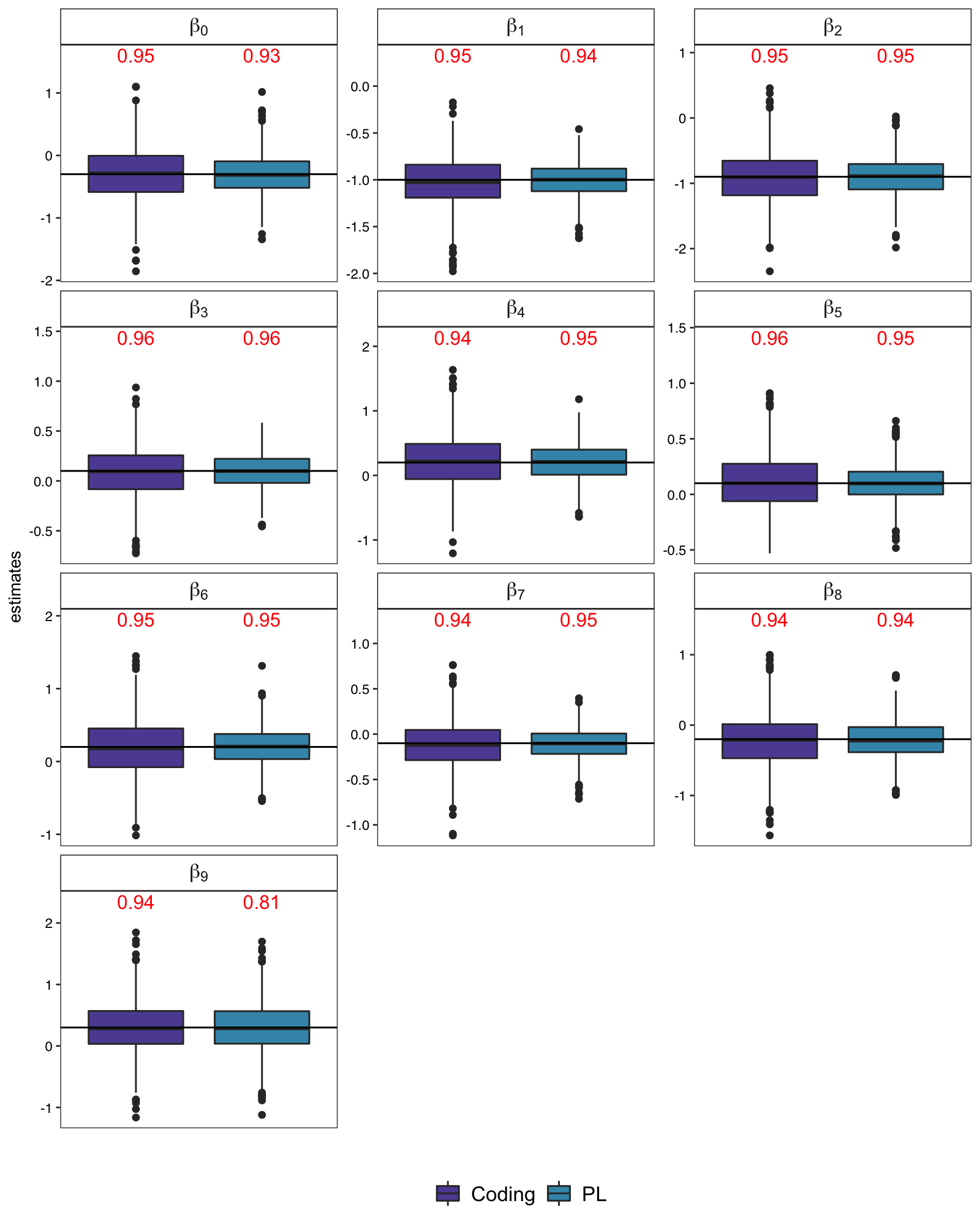}
\end{figure}

\begin{table}[!htbp] \centering
	\caption{Simulation results of coding and pseudo-likelihood based estimators of network causal effects for low density network of size 200} 
	\label{} 
	\begin{tabular}{cccccc}
		\\[-1.8ex]\hline 
		& Truth  & Absolute Bias & MC Variance & Robust Variance & 95\% CI Coverage \\ 
		\hline
		\multicolumn{1}{l}{Coding ($n_1 = 94$)} &&&&& \\
		$E(Y(\mathbf{a}))$ 			& 0.211     & $0.005$   & $0.005$  & $0.0004$ & 0.935 \\ 
		Spillover 						  &  -0.166  & $0.002$   & 0.021  & 0.015 & 0.922 \\ 
		Direct 							   &  -0.179  & $0.006$ & $0.008$ & $0.008$ & 0.944 \\ 
		&&&&& \\
		\multicolumn{1}{l}{Pseudo-likelihood}  &&&&& \\
		$E(Y(\mathbf{a}))$ 			& 0.211    & $0.005$    & $0.002$ & - &  - \\ 
		Spillover 						  &  -0.166 & $0.002$    & 0.010 & - & - \\ 
		Direct 							   &  -0.179 & $0.006$  & $0.004$ & - & - \\ 
		\hline 
	\end{tabular}
	\newline
\end{table}

\begin{table}[!htbp] \centering 
	\caption{Simulation results of coding and pseudo-likelihood based estimators of network causal effects for low density network of size 400} 
	\label{} 
	\begin{tabular}{cccccc}
		\\[-1.8ex]\hline 
		& Truth  & Absolute Bias & MC Variance & Robust Variance & 95\% CI Coverage \\ 
		\hline
		\multicolumn{1}{l}{Coding ($n_1 = 183$)}  &&&&& \\
		$E(Y(\mathbf{a}))$ 			& 0.211       & $0.003$ & $0.002$ & 0.002 & 0.936 \\ 
		Spillover 						  &  -0.166    & $0.006$ & $0.009$ & 0.008 & 0.932 \\ 
		Direct 							   &  -0.179    & $-0.002$ & 0.004 & 0.004 & 0.961 \\ 
		&&&&& \\
		\multicolumn{1}{l}{Pseudo-likelihood}   &&&&& \\
		$E(Y(\mathbf{a}))$ 		& 0.211    & $0.003$ & 0.001 & - & - \\ 
		Spillover 					  &  -0.166 & $0.006$    & 0.004 & - & - \\ 
		Direct 						   &  -0.179 & $-0.002$    & 0.002 & - & - \\ 
		\hline
	\end{tabular}
	\newline
\end{table}

\begin{table}[!htbp] \centering
	\caption{Simulation results of coding and pseudo-likelihood based estimators of network causal effects for low density network of size 1,000} 
	\label{} 
	\begin{tabular}{cccccc}
		\\[-1.8ex]\hline 
		& Truth  &  Bias & MC Variance & Robust Variance & 95\% CI Coverage \\ 
		\hline
		\multicolumn{1}{l}{Coding ($n_1 = 433$)} &&&&& \\
		$E(Y(\mathbf{a}))$ 			& 0.211     & $0.001$   & $0.001$  & $0.001$ & 0.941 \\ 
		Spillover 						  &  -0.167  & $0.001$   & 0.004  & 0.005 & 0.935 \\ 
		Direct 							   &  -0.179  & $0.001$ & $0.002$ & $0.002$ & 0.943 \\ 
		&&&&& \\
		\multicolumn{1}{l}{Pseudo-likelihood}  &&&&& \\
		$E(Y(\mathbf{a}))$ 			& 0.211    & $0.001$    & $<0.001$ & - &  - \\ 
		Spillover 						  &  -0.148 & $0.001$    & 0.002 & - & - \\ 
		Direct 							   &  -0.176 & $0.001$  & $0.001$ & - & - \\ 
		\hline 
	\end{tabular}
	\newline
\end{table}

\pagebreak
\section*{Additional simulation results: small, dense networks} 

\begin{figure}[!htbp] 
	\centering
	\caption{Simulation results of coding and pseudolikelihood based estimators for the causal estimands for high density network of size 100 and 200}
	\subfloat[Dense network of size 100 ($n_1 = 24$; 44\% of iterations did not converge and are excluded)]{\includegraphics[scale=.22]{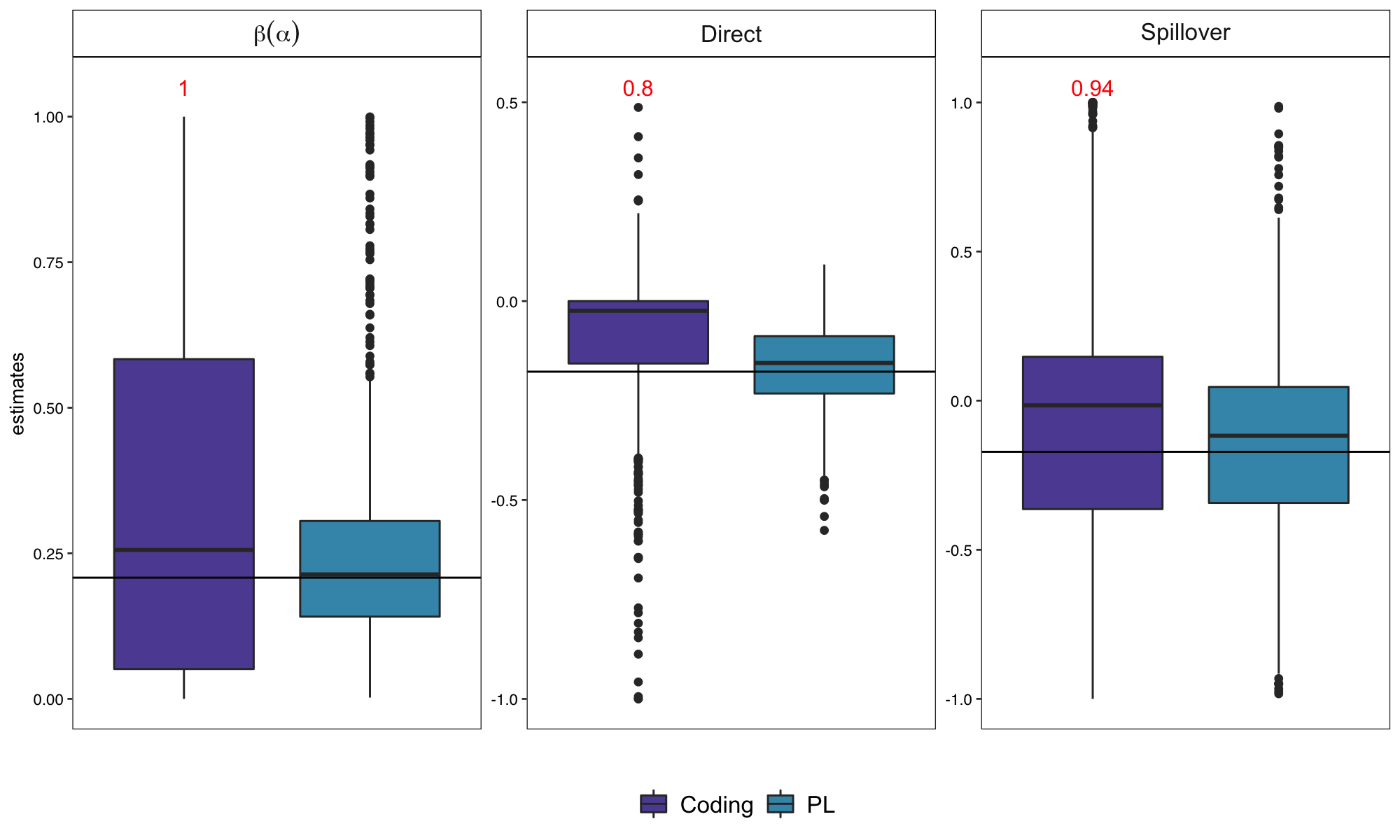}} \\
	\subfloat[Dense network of size 200 ($n_1 = 57$)]{\includegraphics[scale=.22]{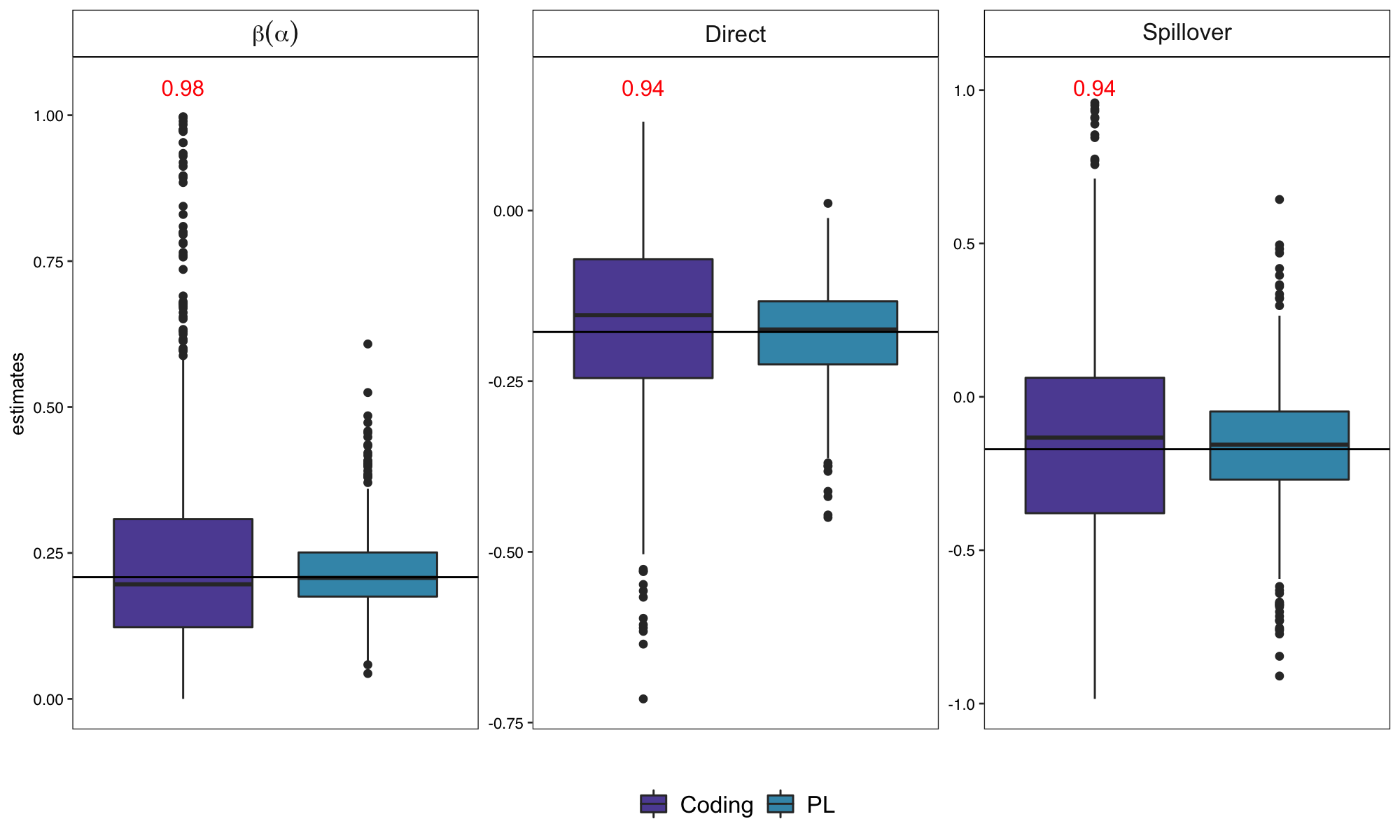}}
\end{figure}

\pagebreak
\section*{Additional simulation results: missing edges} 

\begin{figure}[!htbp]
	\centering
	\caption{Simulation results of coding and pseudo-likelihood based estimators of outcome model gibbs factors for high density network of size 800 with missing edges ($n_1 = 224$)}
	\includegraphics[scale=.28]{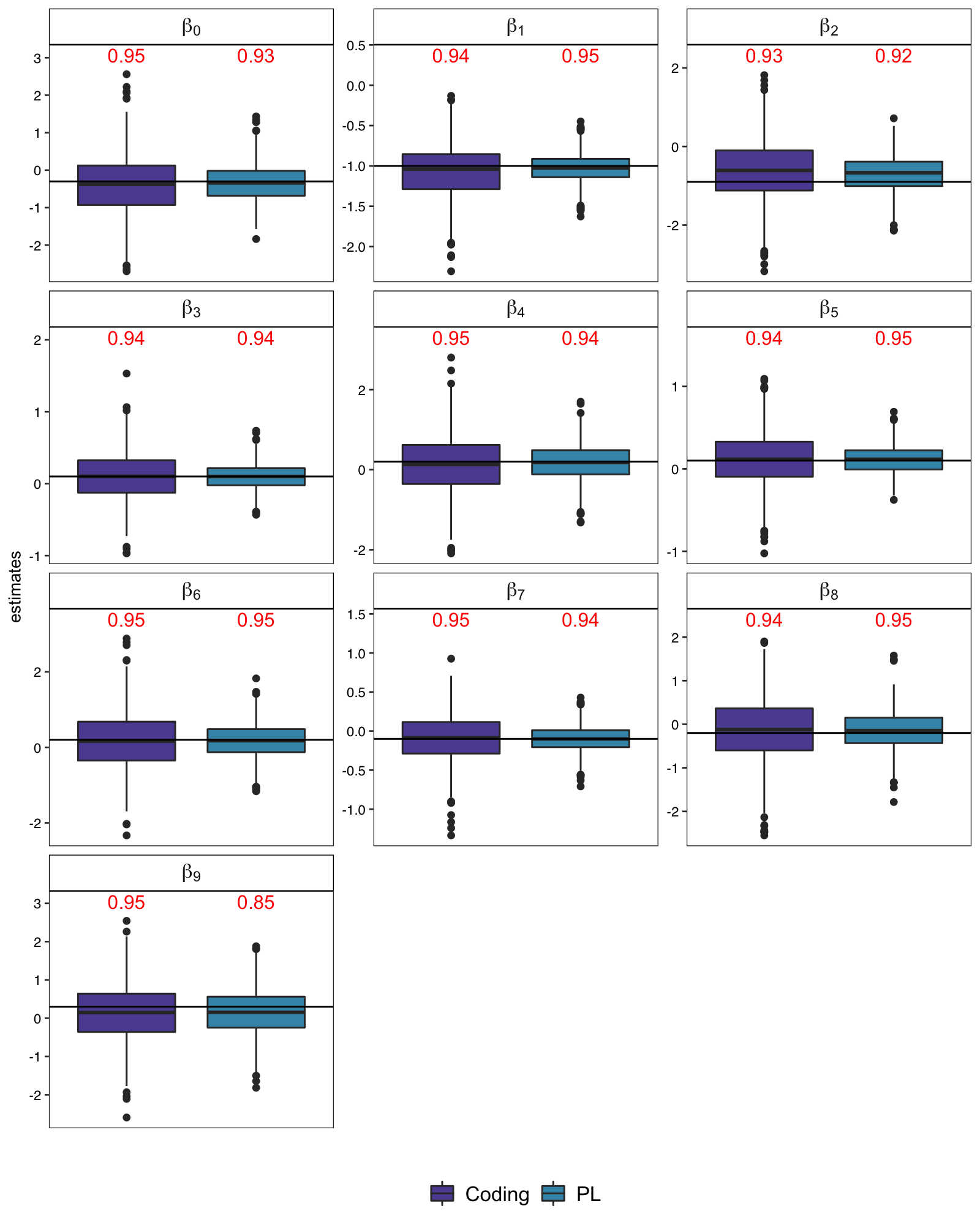}
\end{figure}

\begin{figure}[!htbp] 
	\centering
	\caption{Simulation results of coding and pseudolikelihood based estimators for the causal estimands for high density network of size 800 with and without missing edges}
	\subfloat[No missing edges ($n_1 = 224$)]{\includegraphics[scale=.22]{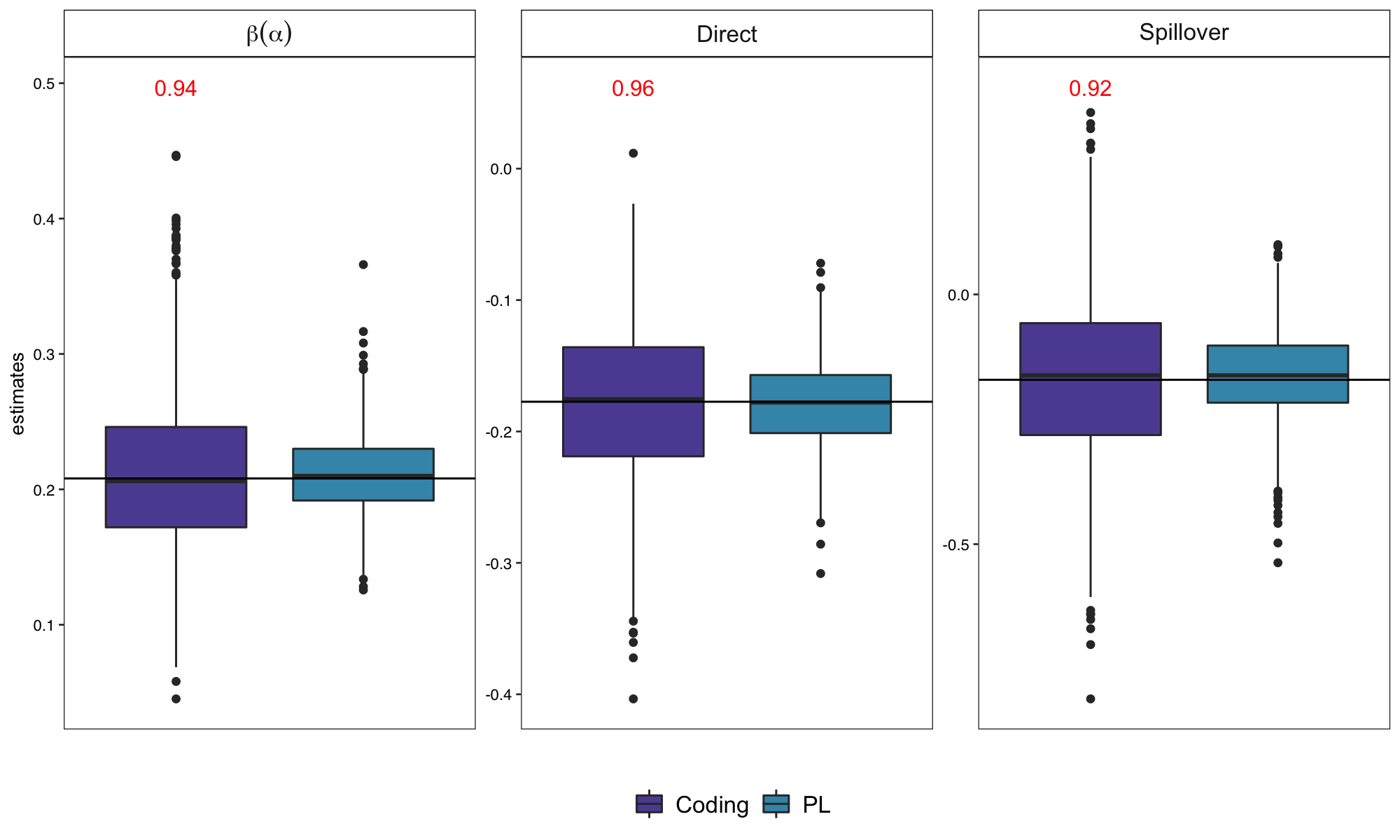}} \\
	\subfloat[Missing edges, 15\% of edges randomly removed ($n_1 = 254$)]{\includegraphics[scale=.22]{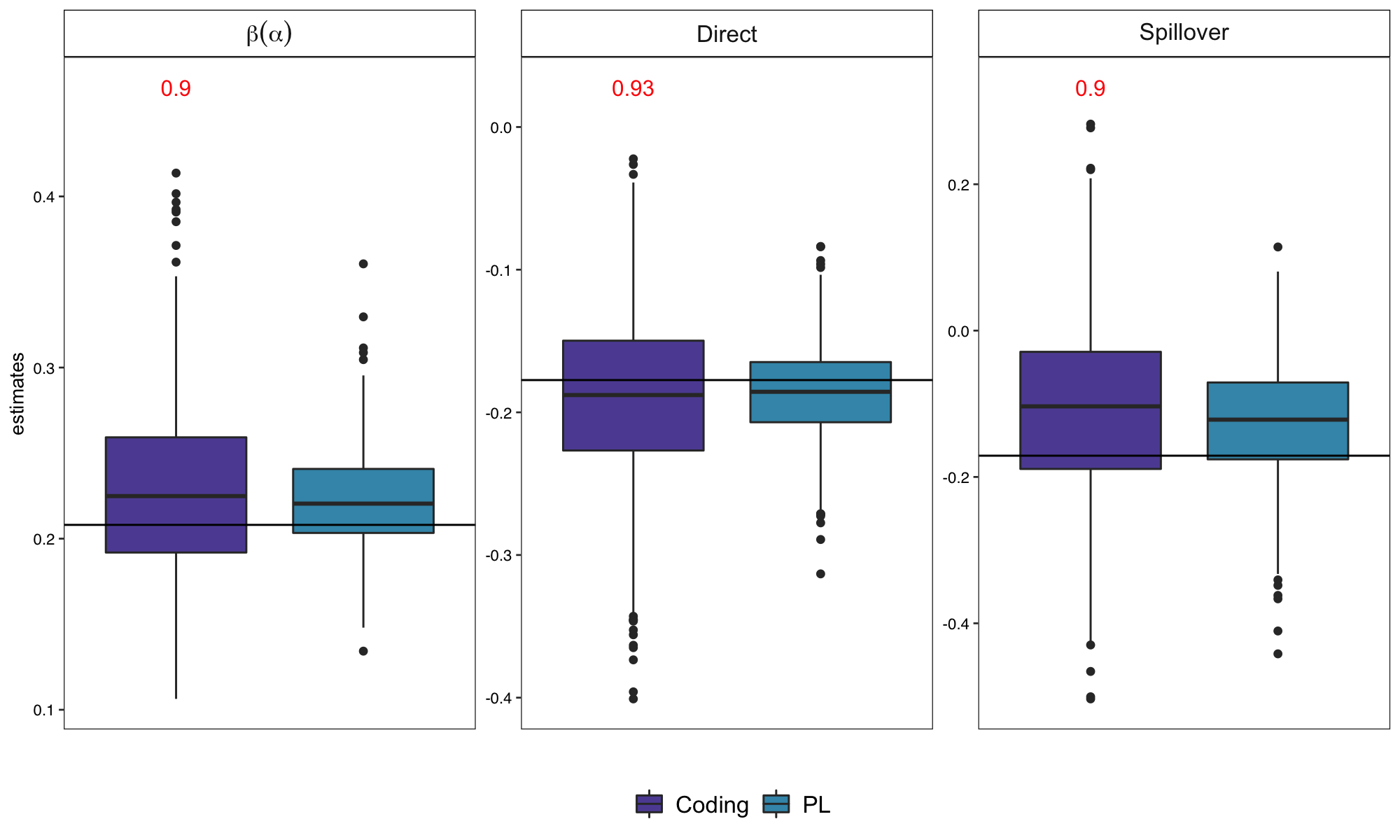}}
\end{figure}

\pagebreak 
\section*{Data application: auto-model parameters}

\begin{figure}[!htbp] 
	\centering
	\caption{Auto-model parameters where the circles give point estimates and line gives the 95\% confidence interval}
	\subfloat[Outcome]{\includegraphics[scale=.10]{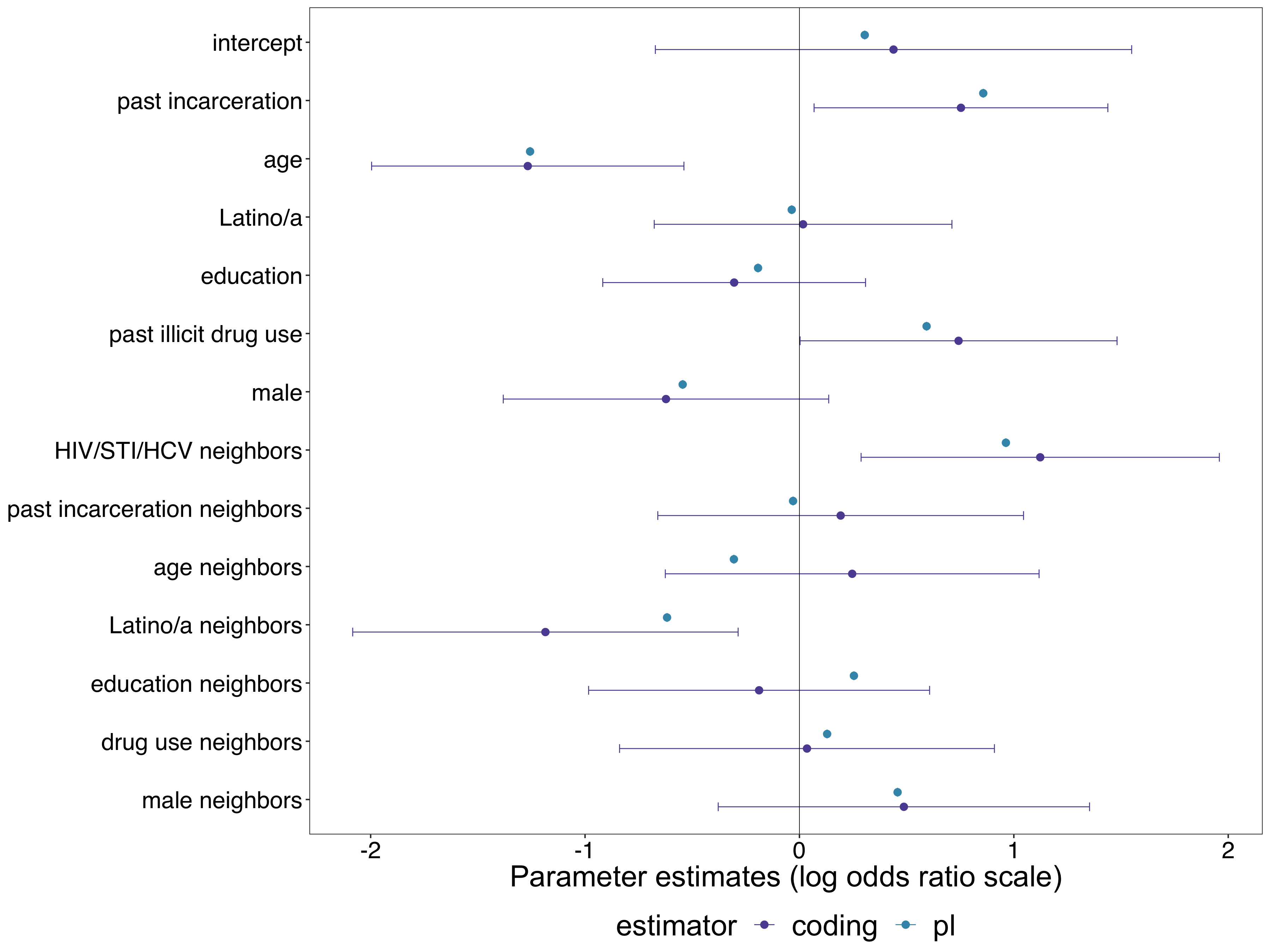}} \\
	\subfloat[Covariate]{\includegraphics[scale=.10]{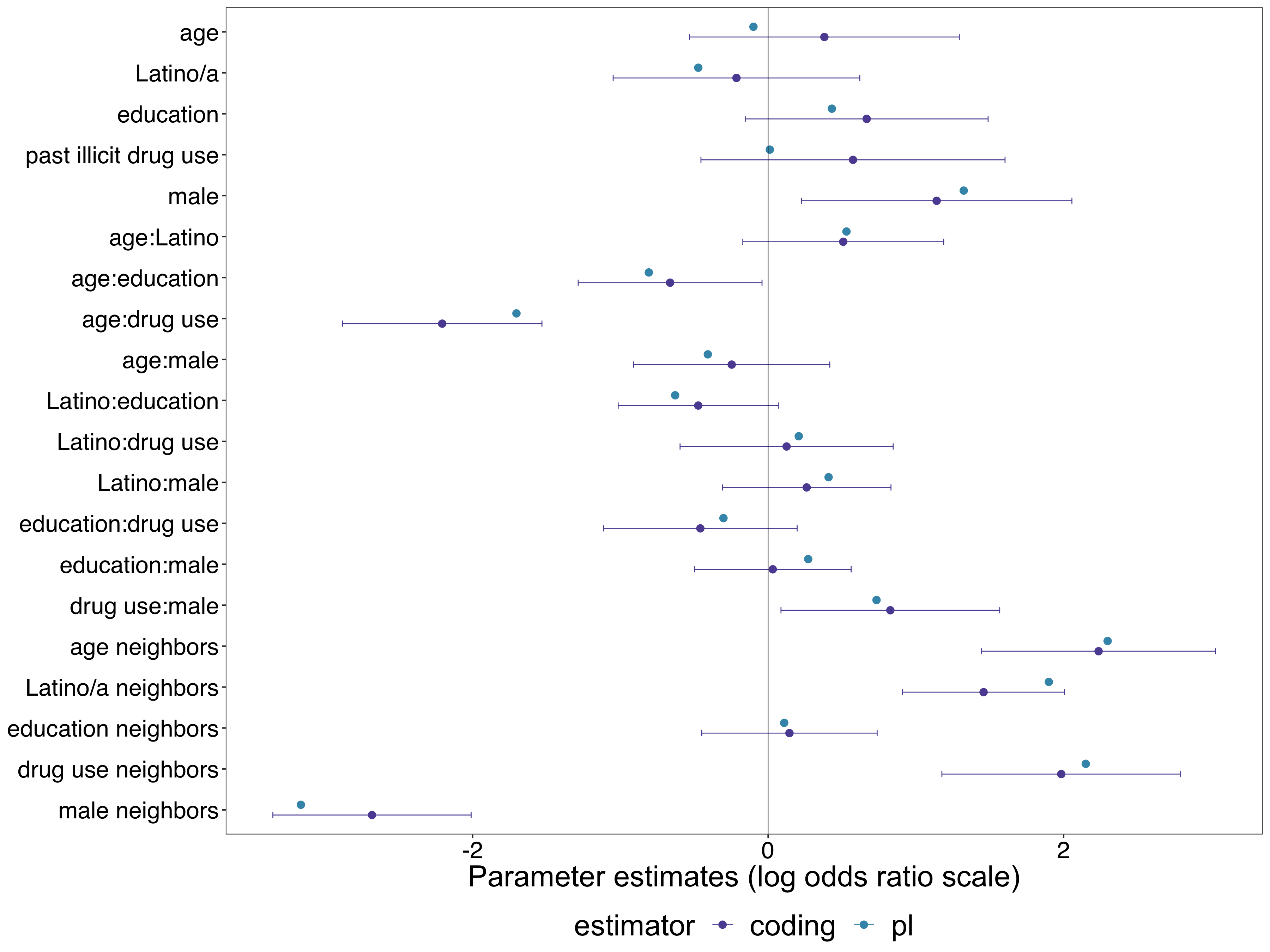}}
\end{figure}

\section*{Data application: alternate model specifications}

We propose the two alternate specifications of the outcome auto-model that incorporate neighbor terms, $W_i$, excluding covariate terms ($L_i$, $\sum_j L_j$). Note that we accounted for covariates in estimation. The auto-model parameter estimates for each option are given in Table 4. The network causal effect estimates are given in Table 5. 

\begin{align}
\textrm{logit}( Pr[Y_i = 1 | W_i, A_i, A_j, Y_i, Y_j ]) & = \beta_0 + \beta_1 A_i + \beta_2 \sum_j A_j/W_i + \beta_3 \sum_j Y_j/W_i + \beta_4 W_i \tag{A} \\
\textrm{logit}( Pr[Y_i = 1 | W_i, A_i, A_j, Y_i, Y_j ]) & = \beta_0 + \beta_1 A_i + \beta_2 \sum_j A_j + \beta_3 \sum_j Y_j + \beta_4 W_i \tag{B} 
\end{align}

\begin{table}[!htbp]
	\caption{Outcome auto-model parameters estimates for coding estimators (excluding covariates) on the odds ratio scale}%
	\centering
	\begin{tabular}
		[c]{lcc}
		&  & \\[-1.8ex]\hline
		& Estimates & 95\% CI \\\hline
		\multicolumn{1}{l}{\textbf{Alternate option (A)}}& &   \\
		\ \ Past incarceration status (individual) & 2.05 & [1.03, 4.11] \\
		\ \ Past incarceration status (proportion neighbors)  & 1.15 & [0.49, 2.69] \\
		\ \ HIV/STI/HCV status (proportion neighbors) & 2.78 & [1.19, 6.50] \\
		\ \ Number of neighbors & 1.28 & [0.84, 1.93] \\
		&  &  \\
		\multicolumn{1}{l}{\textbf{Alternate option (B)}} &  &  \\
		\ \ Past incarceration status (individual) & 2.09 & [1.04, 4.21] \\
		\ \ Past incarceration status (sum neighbors)  & 1.30 & [0.69, 2.44]  \\
		\ \ HIV/STI/HCV status (sum neighbors) & 2.33 & [1.16, 4.70]  \\
		\ \ Number of neighbors & 0.62 & [0.30, 1.31]  \\
		\hline
	\end{tabular}
\end{table}

\begin{table}[!htbp]
	\caption{Network causal effect estimates for coding based estimators (excluding covariates)}%
	\centering
	\begin{tabular}
		[c]{lcc}
		&  & \\[-1.8ex]\hline
		& Estimates & 95\% CI \\\hline
		\multicolumn{1}{l}{\textbf{Alternate option (A)}}& &   \\
		\ \ Spillover &  0.028  & [-0.053, 0.122]  \\
		\ \ Direct  & 0.131 & [0.001, 0.245] \\
		&  &  \\
		\multicolumn{1}{l}{\textbf{Alternate option (B)}} &  &  \\
		\ \ Spillover & 0.121 & [-0.063, 0.243] \\
		\ \ Direct  & 0.150 & [0.016, 0.279]  \\
		\hline
	\end{tabular}
\end{table}

\section*{Data application: simulation under the sharp null} 
We conducted a simulation study to evaluate the operating characteristics of our proposed estimator under the sharp null in the NNAHRAY network. The network structure for the simulation study was based on the NNAHRAY observed network structure excluding singletons (i.e. persons with no ties), $N=412$ ($n_1 = 229$). The true parameter values for $\boldsymbol{\alpha}$ and $\boldsymbol{\beta}$ are given in Table 6. We generated 500 realizations of ($\mathbf{Y},\mathbf{A},\mathbf{L}$) by sampling from the Gibbs factors at their true parameter values. 

\begin{table}[!htbp]
	\centering
	\caption{NNAHRAY simulation parameter values}
	\begin{tabular}{cl|cc|cc}
		\multicolumn{2}{c|}{\textbf{Covariate model $\boldsymbol{\alpha}$}}      & \multicolumn{2}{c|}{\textbf{Treatment model $\boldsymbol{\theta}$}} & \multicolumn{2}{c}{\textbf{Outcome model $\boldsymbol{\beta}$}} \\ \hline
		$\boldsymbol{\tau}$  & \multicolumn{1}{c|}{(0.8,1.0,0.1,-0.1)}           & $\theta_0$                         & -1.5                           & $\beta_0$                        & -1.1                          \\
		$\boldsymbol{\rho}$  & \multicolumn{1}{c|}{(-1.1,-0.3,0.3,-0.7,0.2,1.0)} & $\boldsymbol{\theta_1}$            & (-0.3,-0.6,1.1,1.1)            & $\beta_1$                        & 0                           \\
		$\boldsymbol{\nu}$   & \multicolumn{1}{c|}{(1.0 -0.2,1.0,-1.0)}          & $\theta_2$                         & 0                            & $\boldsymbol{\beta}_2$           & (-0.3,0.0,0.8,-0.3)           \\
		\multicolumn{1}{l}{} &                                                   & $\boldsymbol{\theta_3}$            & (-0.3,0.2,0.9,-0.9)            & $\beta_3$                        & 0                           \\
		\multicolumn{1}{l}{} &                                                   &                                    &                                & $\beta_4$                        & 0                           \\
		\multicolumn{1}{l}{} &                                                   &                                    &                                & $\boldsymbol{\beta}_5$           & (-0.5,0.0,0.5,0.7)           
	\end{tabular}
\end{table}

We then calculated the ``true" value of the direct and spillover effect using the same method described in the Simulation section 5. For all effects, we considered treatment assignment to be a Bernoulli random variable with probability $\gamma = 0.5$. The spillover effects compared network incarceration allocations of 50\% to 0\% and the direct effect was estimated at an incarceration allocation of 50\%. Under the sharp null, the true network direct and spillover effects were 0. The results are as expected and given in Table 7. 

\begin{table}[!htbp] \centering
	\caption{Simulation results of coding estimators of network causal effects in NNAHRAY network under sharp null} 
	\label{} 
	\begin{tabular}{cccccc}
		\\[-1.8ex]\hline 
		& Truth  &  Mean estimate & MC Variance & 95\% CI Coverage \\ 
		\hline
		Spillover 						  &  0  & $-0.0003$   & 0.005 &  0.952 \\ 
		Direct 							   &  0 & $-0.0005$ & 0.002 &  0.955 \\ 
		\hline
	\end{tabular}
	\newline
\end{table}

\end{document}